\documentclass{amsart}

\usepackage{amssymb}
\usepackage[mathscr]{eucal}
\usepackage{graphicx}
\usepackage{mathrsfs}
\usepackage{psfrag}
\usepackage{fullpage}
\usepackage{color, harvard}

\theoremstyle{plain}
\newtheorem{theorem}{Theorem}[section]

\newtheorem{proposition}[theorem]{Proposition}
\newtheorem{lemma}[theorem]{Lemma}

\newtheorem{condition}[theorem]{Condition}

\newtheorem{remark}[theorem]{Remark}

\newcommand{\lb}{\left\{}
\newcommand{\rb}{\right\}}
\newcommand{\Def}{\overset{\text{def}}{=}}

\newcommand{\R}{\mathbb{R}}
\newcommand{\N}{\mathbb{N}}

\newcommand{\KK}{K}

\newcommand{\Borel}{\mathscr{B}}
\newcommand{\Pspace}{\mathscr{P}}
\newcommand{\BP}{\mathbb{P}}
\newcommand{\BE}{\mathbb{E}}
\newcommand{\filt}{\mathscr{F}}
\newcommand{\gilt}{\mathscr{G}}

\newcommand{\la}{\left \langle}
\newcommand{\ra}{\right\rangle}

\newcommand{\ee}{\mathfrak{e}}
\newcommand{\genA}{\mathcal{A}}
\newcommand{\genB}{\mathcal{B}}
\newcommand{\genL}{\mathcal{L}}
\newcommand{\jump}{\mathcal{J}}

\newcommand{\pp}{\mathsf{p}}
\newcommand{\pt}{\star}
\newcommand{\SSS}{\mathcal{S}}

\newcommand{\PP}{\mathcal{P}}

\newcommand{\NN}{{N,n}}
\newcommand{\mart}{\mathcal{M}}
\newcommand{\dfi}{\textsc{m}}
\newcommand{\QQ}{\mathcal{Q}}

\allowdisplaybreaks
\begin{document}
\title{Large Portfolio Asymptotics for Loss From Default}

\author{Kay Giesecke}
\address{Department of Management Science and Engineering\\
Stanford University\\
Stanford, CA 94305}
\email{giesecke@stanford.edu}

\author{Konstantinos Spiliopoulos}
\address{Division of Applied Mathematics\\
Brown University\\
Providence, RI 02912}
\email{kspiliop@dam.brown.edu}

\author{Richard B. Sowers}
\address{Department of Mathematics\\
   University of Illinois at Urbana--Champaign\\
   Urbana, IL 61801}
\email{r-sowers@illinois.edu}

\author{Justin A. Sirignano}
\address{Department of Management Science and Engineering\\
Stanford University\\
Stanford, CA 94305}
\email{jasirign@stanford.edu}


\date{\today. First version: August 15, 2011. We are grateful to two anonymous referees and to Michael Gordy, Monique Jeanblanc and participants of the 5th Financial Risks International Forum in Paris, the 2011 Annual INFORMS Meeting in Charlotte, the 2012 Workshop on Mathematical Finance in Mumbai, the 2011 International Congress on Industrial and Applied Mathematics in Vancouver, the Mathematics Seminars at the University of Southern California and Santa Clara University, and the SEEM Seminar at the Chinese University of Hong Kong for comments.}

\begin{abstract}
We prove a law of large numbers for the loss from default and use it for approximating the distribution of the loss from default in large, potentially heterogenous portfolios. The density of the limiting measure is shown to solve a non-linear stochastic PDE, and certain moments of the limiting measure are shown to satisfy an infinite system of SDEs. The solution to this system leads to the distribution of the limiting portfolio loss, which we propose as an approximation to the loss distribution for a large portfolio. Numerical tests illustrate the accuracy of the approximation, and highlight its computational advantages over a direct Monte Carlo simulation of the original stochastic system.\\
\begin{center}
\normalsize
{\it Mathematical Finance}, forthcoming
\end{center}
\end{abstract}

\maketitle

\section{Introduction}
Reduced-form point process models of correlated default timing are
widely used to measure portfolio credit risk and to value securities
exposed to correlated default risk. Computing the distribution of
the loss from default in these models tends to be difficult,
however, especially in bottom-up formulations with many names.
Semi-analytical transform methods have limited scope. Monte Carlo
simulation is much more broadly applicable  but  can be slow for
the large portfolios and longer time horizons common in practice.

This paper develops an approximation to the distribution of the loss from default in large portfolios
 that may have a heterogenous structure. The approximation is valid for a class of reduced-form models
 in which a name defaults at a stochastic intensity that is influenced by an idiosyncratic risk factor
 process, a systematic risk factor process $X$ common to all names in the pool, and the portfolio loss rate.
 It is based on a law of large numbers for the portfolio loss rate. The limiting portfolio loss is not
 deterministic but follows a stochastic process driven by $X$, indicating that the
exposure to the systematic risk cannot be diversified. We show that the density of the limiting measure, if it exists, satisfies a nonlinear stochastic partial differential equation (SPDE) driven by $X$. We develop a numerical method for solving this equation. The method is based on the observation that certain moments of the limiting measure satisfy an infinite system of SDEs. These SDEs are driven by the systematic risk factor $X$; a truncated system can be solved using a discretization scheme, for example.  The solution to the SDE system leads to the solution to the SPDE through an inverse moment problem. It also leads to the distribution of the limiting portfolio loss, which we propose as an approximation to the distribution of the loss from default for a large portfolio. Estimators of portfolio value at risk and other risk measures are immediate from the limiting loss distribution.

Numerical tests illustrate the accuracy and computational efficiency of the approximation for large but finite portfolios. We find a substantial reduction in computational effort over the alternative of direct Monte Carlo simulation of the high-dimensional original stochastic system. The accuracy of the approximation mainly depends on the portfolio size $N$ and the sensitivity to the systematic risk factor. For a given sensitivity, the accuracy increases with $N$, as expected. The higher the sensitivity to the systematic risk, the higher the variance of the loss distribution and the more accurate is the approximation for fixed $N$. The approximation is remarkably accurate in the tail of the loss distribution. This renders it particularly suitable for the estimation of risk measures for the large pools of loans commonly held by banks. 

Large portfolio approximations were first studied by \citeasnoun{vasicek}. In Vasicek's static model of a homogenous pool, firms default independently of one another conditional on a normally distributed random variable representing a systematic risk factor. Because the losses from defaults are conditionally i.i.d., the classical law of large numbers ensures the convergence of the portfolio loss rate to its conditional mean, from which the limiting loss distribution is immediate. \citeasnoun{schloegl-okane} examine alternative distributions of the systematic factor, and \citeasnoun{lucas-etal} and \citeasnoun{Gordy03arisk-factor} study the limiting loss in a heterogenous portfolio. \citeasnoun{hambly} analyze a dynamic extension of Vasicek's homogeneous pool model in which the systematic risk factor follows a Brownian motion. They obtain an SPDE driven by that Brownian motion for the density of the limiting measure, which they solve using a finite element method. The conditional independence of defaults can also be exploited to analyze the tail behavior of the losses in large, not necessarily homogenous, portfolios using large deviations arguments; see \citeasnoun{ddd} and \citeasnoun{glasserman-kang}.

The analysis in this paper differs from that in the aforementioned articles in several important respects. We study a class of dynamic point process models of correlated default timing in which a firm defaults at a stochastic intensity process. The intensity is influenced by an idiosyncratic risk factor process following a square-root diffusion, a systematic risk factor process following a diffusion with arbitrary coefficient functions, and the portfolio loss rate. To address the heterogeneity of a portfolio, the intensity parameters of each name are allowed to be different. The choice for dependence of the intensity on idiosyncratic and systematic risk factor processes is motivated by the empirical findings of \citeasnoun{duffie-saita-wang}. The choice for dependence of the intensity on the portfolio loss is motivated by the empirical results of \citeasnoun{azizpour-giesecke-schwenkler}, who find that defaults have a statistically significant feedback effect on the surviving firms. The self-exciting behavior of defaults violates the conditional independence property that is exploited in the aforementioned articles. It complicates the asymptotic analysis and induces an integral term in the drift of the SPDE governing the density of the limiting measure. The exposure to the systematic risk leads to the noise term in the SPDE, which is given by an It{\^o} integral against the Brownian motion driving the systematic risk diffusion. The solution to the SPDE governs the distribution of the limiting loss at all future horizons, facilitating the computation of the ``loss surface.'' This dynamic perspective is absent in the static formulations considered in most of the aforementioned articles.

The law of large numbers (LLN) proved in this paper significantly extends an earlier result in \citeasnoun{GieseckeSpiliopoulosSowers2011}, which assumes Ornstein-Uhlenbeck dynamics for the systematic risk factor and the sensitivity of the intensity to the systematic risk to vanish in the large-portfolio limit. The LLN developed here allows for general diffusion dynamics for the systematic risk factor. Moreover, the exposure of an intensity to this diffusion is not required to vanish, generating a much richer, non-deterministic limiting behavior governed by an SPDE rather than a PDE. The treatment of these features requires additional arguments. Our main result (Theorem \ref{T:MainLLN0}) develops the stochastic evolution equation that the limiting empirical measure satisfies. If the limiting empirical measure admits a density, then an integration by parts argument shows that the density satisfies an SPDE. The filtered martingale problem is used to identify the limit and prove the LLN. A major difficulty in the identification of the limit is the solution of a coupled system of SDEs, which we address using fixed-point arguments. In contrast to \citeasnoun{GieseckeSpiliopoulosSowers2011}, this system does not easily decouple in the more general setting considered here. The analysis of the fixed-point arguments is complicated due to the square-root singularity.



\citeasnoun{CMZ} prove a LLN for a related system, taking an intensity as a function of an idiosyncratic risk factor, a systematic risk factor, and the portfolio loss rate. The risk factors follow diffusion processes whose coefficients may depend on the portfolio loss. In that formulation, the impact of a default on the dynamics of the surviving firms is permanent. In our work, an intensity depends on the path of the portfolio loss.  Therefore, the impact of a default on the surviving firms may be transient and fade away with time. There is a recovery effect. Other interacting particle systems with permanent default impact are analyzed by \citeasnoun{daipra-etal} and \citeasnoun{daipra-tolotti}, who take an intensity as a function of the portfolio loss rate. In a model with local interaction, \citeasnoun{giesecke-weber} take the intensity of a name as a function of the state of the names in a specified neighborhood of that name. These papers prove LLNs for the portfolio loss and develop Gaussian approximations to the portfolio loss distribution based on central limit theorems.  The interacting particle system which we propose and study includes firm-specific sources of default risk and addresses an additional source of default clustering, namely the exposure of a firm to a systematic risk factor process. This exposure generates a random limiting behavior. 

There are several other related articles. \citeasnoun{davis-rod} develop large portfolio approximations based on a law of large numbers and a central limit theorem in a stochastic network setting, in which firms default independently of one another conditional on the realization of a systematic factor governed by a finite state Markov chain. \citeasnoun{sircar-zari} examine large portfolio asymptotics for utility indifference valuation of securities exposed to the losses in the pool. Our formulation addresses the dependence of the intensity on a systematic diffusion factor and the portfolio loss. It allows for self-exciting effects that violate the conditional independence assumption. Gaussian and large deviation approximations to the distribution of portfolio losses for an affine point process system with features similar to those of our, not necessarily affine, system are provided by \citeasnoun{zhang-etal}. Their asymptotic analysis is based on a ``large horizon'' rather than a ``large portfolio'' regime considered here and in the aforementioned articles. The scope of their approximations differs from that of ours.


The rest of the paper is organized as follows. Section \ref{S:Model} describes a class of point process models of correlated default timing in a pool of names. Section \ref{S:Typical} states our main result, Theorem \ref{T:MainLLN0}, a law of large numbers for the loss rate in the pool. Section  \ref{S:HomogeneousPool} provides further insights into the limiting behavior of the loss for the special case of a homogenous pool. Section \ref{S:Numerics} develops, implements and tests a moment method for the numerical solution of the SPDE. Numerical results illustrate the method and demonstrate the accuracy of the approximation of the portfolio loss by the limiting loss. Section \ref{S:Extensions} discusses the extension of our results to more general intensity dynamics. Sections \ref{S:LimitIdentification} and \ref{S:MainProof} are devoted to the proof of Theorem \ref{T:MainLLN0}. Appendices provide auxiliary results.

\section{Model and Assumptions}\label{S:Model}
We provide a dynamic point process model of correlated default timing in a portfolio of names.  We assume that $(\Omega,\filt,\BP)$ is an underlying probability space on which all random variables are defined. Let $\{W^n\}_{n\in \N}$ be a countable collection of independent standard Brownian motions.  Let $\{\ee_n\}_{n\in \N}$ be an i.i.d. collection of standard exponential random variables which are independent of the $W^n$'s.  Finally, let $V$ be a standard Brownian motion which is independent of the $W^n$'s and $\ee_n$'s.  Each $W^n$ will represent a source of risk which is idiosyncratic to a specific name.  Each $\ee_n$ will represent a normalized default time for a specific name.  The process $V$ will drive a systematic risk factor process to which all names are exposed. Define $\mathcal{V}_{t}=\sigma\left(V_{s}, 0\leq s\leq t\right)\vee \mathcal{N}$
and $\filt_t=\sigma\left( (V_{s},W_{s}^{n}), 0\leq s\leq t, n\in \N\right) \vee \mathcal{N}$, where
$\mathcal{N}$ contains the $\BP$-null sets.

Fix $N\in\N$, $n\in\{1,2,\ldots, N\}$ and consider the following system:
\begin{equation} \label{E:main}
\begin{aligned}
d\lambda^\NN_t &= -\alpha_\NN (\lambda^\NN_t-\bar \lambda_\NN)dt + \sigma_\NN \sqrt{\lambda^\NN_t}dW^n_t +  \beta^C_\NN dL^N_t+ \beta^S_\NN \lambda^\NN_t dX_t \qquad t>0\\
\lambda^\NN_0 &= \lambda_{\circ, N,n}\\
dX_t &= b_{0}( X_t) dt + \sigma_{0}(X_t)dV_t \qquad t>0\\
X_0&= x_\circ \\
L^N_t &= \frac{1}{N}\sum_{n=1}^N \chi_{[\ee_n,\infty)}\left(\int_{s=0}^t \lambda^\NN_s ds\right). \end{aligned}
\end{equation}
Here, $\chi$ is the indicator function. The initial condition $x_\circ$ of $X$ is fixed.
The $\alpha_\NN,\bar \lambda_\NN,\sigma_\NN,\beta^C_\NN,\beta^S_\NN$ are constant parameters, for each $N$ and $n$. We discuss their meaning below. 
The description of $L^N$ is equivalent to a more standard construction. In particular,  define
\begin{equation}\label{E:tau}
\tau^\NN \Def \inf\lb t\ge 0: \int_{s=0}^t \lambda^\NN_s ds\ge \ee_n\rb.
\end{equation}
Then $\chi_{[\ee_n,\infty)}(\int_{s=0}^t \lambda^\NN_s ds) = \chi_{\{\tau^\NN\le t\}}$
and consequently
\begin{equation} L^N_t =\frac{1}{N}\sum_{n=1}^N \chi_{\{\tau^\NN\le t\}}. \end{equation}

The process $L^N$ represents the loss rate in a portfolio of $N$ names, assuming a loss given default of one unit. The process $\lambda^\NN$ represents the intensity, or conditional event rate, of the $n$-th name in the pool.  More precisely, $\lambda^\NN$ is the density of the Doob-Meyer compensator to the default indicator $\chi_{\{\tau^\NN\le t\}}$; see \eqref{E:DoobMeyer}. The results in Section 3 of \citeasnoun{GieseckeSpiliopoulosSowers2011} imply that the system (\ref{E:main}) has a unique  solution such that $\lambda^\NN_t\ge 0$ for every $N\in\N$, $n\in\{1,2,\ldots, N\}$ and $t\ge 0$. Thus, the model is well-posed.

The jump-diffusion intensity model (\ref{E:main}) is empirically motivated. It addresses several channels of default clustering. An intensity is driven by an idiosyncratic source of risk represented by a Brownian motion $W^n$, and a source of systematic risk common to all firms--the diffusion process $X$.  Movements in $X$ cause correlated changes in firms' intensities and thus provide a channel for default clustering emphasized by \citeasnoun{ddk} for corporate defaults in the U.S.  The sensitivity of $\lambda^\NN$ to changes in $X$ is measured by the parameter $\beta^S_\NN\in\R$.  The second channel for default clustering is modeled through the feedback (``contagion'') term $\beta^C_\NN dL^N_t$.  A default causes a jump of size $\tfrac1N\beta^C_\NN$ in the intensity $\lambda^\NN$, where $\beta^C_\NN\in \R_+= [0,\infty)$.  Due to the mean-reversion of $\lambda^\NN$, the impact of a default fades away with time, exponentially with rate $\alpha_\NN\in\R_+$. \citeasnoun{azizpour-giesecke-schwenkler} have found self-exciting effects of this type to be an important channel for the clustering of defaults in the U.S., over and above any clustering caused by the exposure of firms to systematic risk factors. \citeasnoun{giesecke-schwenkler} develop and analyze likelihood estimators of the parameters of point process models such as (\ref{E:main}).


We allow for a heterogeneous pool; the intensity dynamics of each name can be different.  We capture these different dynamics by defining the ``types''
\begin{equation}\label{E:typedef} \pp^\NN \Def (\alpha_\NN,\bar \lambda_\NN,\sigma_\NN,\beta^C_\NN,\beta^S_\NN); \end{equation}
the $\pp^\NN$'s take values in parameter space $\PP\Def \R_+^4\times \R$.  In order to expect regular macroscopic behavior of $L^N$ as $N\to \infty$, the $\pp^\NN$'s and
the $\lambda_{\circ,N,n}$'s should have enough regularity as $N\to \infty$.
For each $N\in \N$, define
\begin{equation*} \pi^N \Def \frac{1}{N}\sum_{n=1}^N \delta_{\pp^\NN} \qquad \text{and}\qquad \Lambda^N_\circ \Def \frac{1}{N}\sum_{n=1}^N \delta_{\lambda_{\circ,N,n}}; \end{equation*}
these are elements of $\Pspace(\PP)$ and $\Pspace(\R_+)$ respectively\footnote{As usual, if $E$ is a topological space, $\Pspace(E)$ is the collection of Borel probability measures on $E$.}.

We require three main conditions. These conditions are in force throughout the paper, even though this may not always be stated explicitly.
Firstly, we assume that the types of \eqref{E:typedef} and the initial distributions (the $\lambda_{\circ,N,n}$'s) are sufficiently regular.
\begin{condition}\label{A:regularity}  $\pi \Def \lim_{N\to \infty}\pi^N$ and $\Lambda_\circ\Def \lim_{N\to \infty}\Lambda^N_\circ$ exist \textup{(}in $\Pspace(\PP)$ and $\Pspace(\R_+)$, respectively\textup{)}.
 \end{condition}
We also  require that the $\pi^N$'s and $\Lambda^N_\circ$'s all (uniformly in $N$) have compact support. We could relax this requirement, at the cost of a much more careful error analysis.
\begin{condition}\label{A:Bounded} There is a $\KK>0$ such that the $\alpha_\NN$'s, $\bar \lambda_\NN$'s, $\sigma_\NN$'s,
$\beta^C_\NN$'s, $|\beta^S_\NN|$'s, and $\lambda_{\circ,N,n}$'s are all bounded by  $\KK$ for all $N\in \N$ and $n\in \{1,2,\dots, N\}$. \end{condition}
Regarding the systematic risk process $X$, we assume
\begin{condition}\label{A:RegularityExogenous} The functions $b_{0}$ and $\sigma_{0}$ that govern the systematic risk diffusion $X$ are such that the corresponding
SDE has a unique strong solution. Moreover, there is a function $u(x)$  such that $\sigma_{0}(x)u(x)=-b_{0}(x)$ and for every $T>0$ we have
\begin{equation}
 \BE\left[e^{\frac{1}{2}\int_{0}^{T}|u(X_{s})|^{2}ds}\right]<\infty.\label{Eq:NovikovCondition}
\end{equation}
\end{condition}
The Novikov condition (\ref{Eq:NovikovCondition}) may not be necessary.  Lemma \ref{L:bQDef} is the key step for the proof of a law of large numbers for the loss rate $L^N$ in the system (\ref{E:main}), which is stated as Theorem \ref{T:MainLLN0} below. Its proof is based on a fixed point argument and uses Girsanov's theorem; this is where (\ref{Eq:NovikovCondition}) is required.
\footnote{If condition (\ref{Eq:NovikovCondition}) is required to hold only for some $T>0$, then the statement of Theorem \ref{T:MainLLN0} below will hold
for $t\in[0,T]$ instead of $t>0$.}

Our basic formulation significantly extends that of \citeasnoun{GieseckeSpiliopoulosSowers2011}. First, we allow the systematic risk $X$ to follow a general diffusion process with coefficients satisfying Condition \ref{A:RegularityExogenous} rather than a simple Ornstein-Uhlenbeck process. Second, we no longer require the exposure to the systematic risk, $\beta^S_\NN$, to vanish in the limit as $N\to\infty$. This implies a richer, non-deterministic limiting behavior. The analysis of this behavior is more challenging and requires new arguments.

Section \ref{S:Extensions} discusses further extensions of our basic formulation, including stochastic position losses and more general intensity dynamics.

\section{Law of Large Numbers}\label{S:Typical}
We develop a law of large numbers for the portfolio loss rate $L^N$ in the system (\ref{E:main}). To this end, we need to understand a system which contains a bit more information than the loss rate  $L^N$.   For each $N\in \N$ and $n\in \{1,2,\dots, N\}$, define
\begin{equation}\label{E:DoobMeyer} \dfi^\NN_t \Def \chi_{[0,\ee_n)}\left(\int_{s=0}^t \lambda^\NN_s ds\right) = \chi_{\{\tau^\NN>t\}} \end{equation}
(where $\tau^\NN$ is as in \eqref{E:tau}).  In other words, $\dfi^\NN_t=1$ if and only if the $n$-th name is still alive at time $t$; otherwise $\dfi^\NN_t=0$.
Thus $\dfi^\NN$ is nonincreasing and right-continuous.
It is easy to see that
\begin{equation*} \dfi^\NN_t + \int_{s=0}^t \lambda^\NN_s \dfi^\NN_s ds \end{equation*}
is a martingale.
Define $\hat \PP\Def \PP\times \R_+$.  For each $N\in \N$, define $\hat \pp^\NN_t \Def (\pp^\NN,\lambda^\NN_t)$
for all $n\in \{1,2,\dots, N\}$ and $t\ge 0$.  For each $t\ge 0$, define
\begin{equation*} \mu^N_t \Def \frac{1}{N}\sum_{n=1}^N\delta_{\hat \pp^\NN_t}\dfi^\NN_t; \end{equation*}
in other words we keep track of the empirical distribution of the type and intensity for those assets which are still ``alive''.
We note that
\begin{equation*} L^N_t =1-\mu^N_t(\hat \PP),\quad t\ge 0.\end{equation*}

We want to understand the dynamics of $\mu^N_t$ for large $N$ (this will then imply the ``typical" behavior for $L^N_t$).
To understand what our main result is, let's first set up a topological framework to understand convergence of $\mu^N$.  Let $E$ be the collection of sub-probability measures (i.e., defective probability measures) on $\hat \PP$; i.e., $E$ consists
of those Borel measures $\nu$ on $\hat \PP$ such that $\nu(\hat \PP)\le 1$.
We can topologize $E$ in the usual way (by projecting onto the one-point compactification of $\hat \PP$; see \citeasnoun[Ch. 9.5]{MR90g:00004}). In particular, fix a point $\pt$ that is not in $\hat \PP$ and define $\hat \PP^+\Def \hat \PP\cup \{\pt\}$.
Give $\hat \PP^+$ the standard topology; open sets are those which are open subsets of $\hat \PP$ (with its original topology) or complements in $\hat \PP^+$ of closed subsets
of $\hat \PP$ (again, in the original topology of $\hat \PP$).  Define a bijection $\iota$ from $E$ to $\Pspace(\hat \PP^+)$ (the collection of Borel probability measures
on $\hat \PP^+$) by setting
\begin{equation*} (\iota \nu)(A) \Def \nu(A\cap \hat \PP) + \left(1-\nu(\hat \PP)\right)\delta_{\pt}(A) \end{equation*}
for all $A\in \Borel(\hat \PP^+)$.  We can define the Skorohod topology on $\Pspace(\hat \PP^+)$, and define a corresponding metric on $E$ by requiring $\iota$ to be an isometry.  This makes $E$ a Polish space.
Thus, $\mu^N$ is an element\footnote{If $S$ is a Polish space, then $D_S[0,\infty)$ is the collection of maps from $[0,\infty)$ into $S$ which
are right-continuous and which have left-hand limits.  The space $D_S[0,\infty)$ can be topologized by the Skorohod metric, which we will denote by $d_S$; see \citeasnoun{MR88a:60130}.} of $D_E[0,\infty)$.

The main result of this paper is Theorem \ref{T:MainLLN0},
essentially a law of large numbers for $\mu^N_t$ as
$N\uparrow\infty$.  For $\hat \pp=(\pp,\lambda)$ where
$\pp=(\alpha,\bar \lambda,\sigma,\beta^C,\beta^S)\in \PP$ and $f\in
C^\infty(\hat \PP)$, define the operators
\begin{equation}\label{E:Operators1}
\begin{aligned} (\genL_1 f)(\hat \pp) &= \frac12 \sigma^{2}\lambda\frac{\partial^2 f}{\partial \lambda^2}(\hat \pp) - \alpha(\lambda-\bar \lambda)\frac{\partial f}{\partial \lambda}(\hat \pp)-\lambda f(\hat \pp)\\
(\genL_2 f)(\hat \pp) &= \beta^C \frac{\partial f}{\partial \lambda}(\hat \pp)\\
(\genL_3^{x} f)(\hat \pp) &= \beta^{S}\lambda b_{0}(x)\frac{\partial f}{\partial \lambda}(\hat \pp)+\frac{1}{2}(\beta^{S})^{2}\lambda^{2}\sigma_{0}^{2}(x)\frac{\partial^{2}f}{\partial \lambda^{2}}(\hat \pp)\\
(\genL_4^{x} f)(\hat \pp) &=
\beta^{S}\lambda\sigma_{0}(x)\frac{\partial f}{\partial
\lambda}(\hat \pp).
 \end{aligned}
\end{equation}
Also define
\begin{equation*} \QQ(\hat \pp) \Def \lambda .\end{equation*}
The generator $\genL_1$ corresponds to the diffusive part of the
intensity with killing rate $\lambda$, and $\genL_2$ is the
macroscopic effect of contagion on the surviving intensities at any
given time.  Operators $\genL_3^{x}$ and $\genL_4^{x}$ are related
to the exogenous systematic risk $X$.

For
every $f\in C^\infty(\hat \PP)$ and $\mu\in E$, define
\begin{equation*} \la f,\mu\ra_E \Def \int_{\hat \pp \in \hat \PP}f(\hat \pp)\mu(d\hat \pp). \end{equation*}

\begin{theorem}\label{T:MainLLN0}
We have that $\mu^{N}_{\cdot}$ converges in distribution to $\bar{\mu}_{\cdot}$ in
$D_E[0,T]$. The evolution of
$\bar{\mu}_{\cdot}$ is given by the measure evolution equation
\begin{align*}
d\la f,\bar \mu_t\ra_E &=  \left\{\la \genL_1f,\bar \mu_t\ra_E+ \la \QQ,\bar \mu_t\ra_E
\la \genL_2f,\bar \mu_t\ra_E+\la \genL^{X_{t}}_3 f,\bar
\mu_t\ra_E\right\}dt\nonumber\\
&+\la \genL^{X_{t}}_4 f,\bar \mu_t\ra_E dV_{t},\quad
\forall f\in C^\infty(\hat \PP) \text{ a.s.}
\end{align*}
Suppose there is a solution of the nonlinear SPDE
\begin{align} \label{Eq:NonlinearSPDE}
\begin{aligned}
d\upsilon(t,\hat \pp) &= \left\{\genL_1^*\upsilon(t,\hat \pp) +\genL_3^{*,X_{t}}\upsilon(t,\hat \pp) + \left(\int_{\hat \pp'\in \hat \PP} \QQ(\hat \pp')\upsilon(t,\hat \pp')d\hat \pp'\right) \genL_2^* \upsilon(t,\hat \pp)\right\}dt\\
&+ \genL_4^{*,X_{t}}\upsilon(t,\hat \pp)dV_t,\quad t>0,\quad \hat \pp\in \hat \PP
\end{aligned}
\end{align}
where $\genL_i^*$ denote adjoint operators, with initial condition
\begin{equation*} \lim_{t\searrow 0}\upsilon(t,\hat \pp)d\hat \pp = \pi \times \Lambda_\circ. \end{equation*}
Then
\begin{equation*}
 \bar{\mu}_t = \upsilon(t,\hat \pp)d\hat \pp.
 \end{equation*}
\end{theorem}

\begin{remark}
The SPDE (\ref{Eq:NonlinearSPDE}) should be supplied with appropriate boundary conditions. In Section \ref{SS:BoundaryConditions} below, we will justify the conditions
$$\upsilon(t,\lambda = 0, \pp) = \upsilon(t, \lambda = \infty, \pp) = 0.$$
\end{remark}

The proof of Theorem \ref{T:MainLLN0} is given in Sections \ref{S:LimitIdentification} and \ref{S:MainProof}. Lemma \ref{L:Qchar} provides an alternative characterization of the limit $\bar{\mu}$. Auxiliary results are given in the appendices.

\begin{remark}\label{R:SPDE}
Equation (\ref{Eq:NonlinearSPDE}) is a stochastic partial integral differential equation (SPIDE)  in the half line that degenerates at the boundary $\lambda=0$. Due to Lemma \ref{L:bQDef} below, (\ref{Eq:NonlinearSPDE}) can be viewed as a linear SPDE in the half line that degenerates at the boundary. Indeed, by Lemma \ref{L:bQDef}, there is a unique pair $\{(Q(t),\lambda_{t}(\hat \pp)):t\geq 0\}$ taking values in
$\R_+\times\R_{+}$ satisfying the coupled system (\ref{E:bQDef})-(\ref{E:EffectiveEquation1}).  Noting that $Q(t)=\la \QQ,\bar \mu_t\ra_E$, we see that the SPDE (\ref{Eq:NonlinearSPDE}) can be written as
\begin{equation} d\upsilon(t,\hat \pp) = \left\{\genL_1^*\upsilon(t,\hat \pp) +\genL_3^{*,X_{t}}\upsilon(t,\hat \pp) + Q(t) \genL_2^* \upsilon(t,\hat \pp)\right\}dt + \genL_4^{*,X_{t}}\upsilon(t,\hat \pp)dV_t,\quad t>0,\quad \hat \pp\in \hat \PP. \label{Eq:NonlinearSPDE2}
\end{equation}
Notice also that by Remark \ref{R:QisBounded}, $Q(t)$  is bounded for all $t\in\R_{+}$.
Linear SPDEs in the half line that degenerate at the boundary are treated, under alternative assumptions, by \citeasnoun{KryLot1998}, \citeasnoun{Kim2009}, \citeasnoun{Lot2001} and \citeasnoun{Kim2008}.
\end{remark}

\section{Homogeneous Pool}\label{S:HomogeneousPool}
We develop further insights into the SPDE governing the limit density (if it exists) in the case that the portfolio is homogenous.
Let  $\hat \pp=(\pp,\lambda)$ where $\pp=(\alpha,\bar \lambda,\sigma,\beta^C,\beta^S)\in \PP$. For a homogenous pool, $\hat
\pp^\NN = \hat \pp$ for all $N\in \N$ and $n\in \{1,2,\dots, N\}$. In this case, we write $\upsilon(t,\lambda)$ for the solution of the SPDE (\ref{Eq:NonlinearSPDE}), suppressing the dependence on the fixed $\hat \pp$.

\subsection{Limiting Portfolio Loss} The SPDE takes the form
\begin{equation}\label{Eq:NonlinearSPDEhom}
\begin{aligned}
d\upsilon(t,\lambda) &= \left\{\genL_1^*\upsilon(t,\lambda) +\genL_3^{*,X_{t}}\upsilon(t,\lambda) + \left(\int_{0}^{\infty} \lambda\upsilon(t,\lambda)d\lambda\right) \genL_2^* \upsilon(t,\lambda)\right\}dt + \genL_4^{*,X_{t}}\upsilon(t,\lambda)dV_t,\, t,\lambda>0\\
 \upsilon(0,\lambda) &=   \Lambda_\circ(\lambda),\\
 \upsilon(t,0) &=\lim_{\lambda\nearrow\infty}\upsilon(t,\lambda) =0
\end{aligned}
\end{equation}
where the adjoint operators are given by
\begin{align}
\begin{aligned}
\genL_1^*\upsilon(t,\lambda)&=\frac{\partial^{2}}{\partial\lambda^{2}}\left(\frac12 \sigma^{2}\lambda\upsilon(t,\lambda)\right) +\frac{\partial}{\partial\lambda}\left( \alpha(\lambda-\bar \lambda)\upsilon(t,\lambda)\right)-\lambda \upsilon(t,\lambda)\\
\genL_2^*\upsilon(t,\lambda)&=- \beta^C \frac{\partial \upsilon(t,\lambda)}{\partial \lambda}\\
\genL_3^{*,x}\upsilon(t,\lambda)&=\frac{\partial^{2}}{\partial\lambda^{2}}\left(\frac{1}{2}(\beta^{S})^{2}\lambda^{2}\sigma_{0}^{2}(x)\upsilon(t,\lambda)\right) -\frac{\partial}{\partial\lambda}\left( \beta^{S}\lambda b_{0}(x)\upsilon(t,\lambda)\right)\\
\genL_4^{*,x}\upsilon(t,\lambda)&=-\beta^{S} \sigma_{0}(x)\frac{\partial}{\partial \lambda}\left(\lambda\upsilon(t,\lambda)\right).
\end{aligned}
\end{align}

Define the \emph{limiting portfolio loss} $L$ by
\begin{equation}\label{Eq:LimitDefaultsHom}
L_{t}\Def 1-\int_{0}^{\infty}\upsilon(t,\lambda)d\lambda,\quad t\ge 0;
\end{equation}
this is a random quantity since $\upsilon(t,\lambda)$ depends on the systematic risk $X_{t}$. For large $N$, Theorem \ref{T:MainLLN0} suggests the ``large-portfolio approximation''
\begin{equation}\label{limiting-loss-approx}
L_{t}^{N}\approx L_t,\quad t\ge 0.
\end{equation}

\subsection{Justification of Boundary Conditions}\label{SS:BoundaryConditions}
Note that  $\int_0^{\infty} \upsilon(t, \lambda) d \lambda \leq 1$ and $\upsilon(t, \lambda) \geq 0$ for all $t \geq 0$ and $\lambda \in \mathbb{R}^{+}$.
Assuming that $\upsilon(t, \lambda)$ is continuous in $\lambda$, this gives   $\lim_{\lambda \to \infty} \upsilon(t, \lambda) = 0$,
for otherwise the integral $\int_0^{\infty} \upsilon(t, \lambda) d \lambda$ diverges.

The boundary condition of $\upsilon(t, \lambda = 0)$ is implied by the intensities $\lambda_t^{N,n}$ being positive almost surely. For the deterministic case of $\beta^S = 0$, it is sufficient to stipulate Feller's condition of $\alpha \bar{\lambda} > \frac{1}{2} \sigma^2$, see \citeasnoun{FellerTwoSingularDiffusionProblems}.  For the case of $\beta^S > 0$, let us assume that $\upsilon(t, \lambda = 0) < \infty$.  Then, Feller's condition is again sufficient to imply $\upsilon(t, \lambda = 0) = 0$.   Let the flux be  $f(t)$ where $f(t) \Def\frac{\partial}{\partial t} \int_0^{\infty} \upsilon(t, \lambda) d \lambda = - \int_0^{\infty} \lambda \upsilon (t, \lambda) d \lambda$.  This follows from the fact that the $\lambda_t^{N,n}$ stay non-negative almost surely by Lemma $3.1$  in \citeasnoun{GieseckeSpiliopoulosSowers2011} and therefore, according to the empirical measure $\mu_t^N$, only leave $[0, \infty)$ by defaulting.  In the asymptotic case described by the SPDE, defaults occur via the sink term $- \lambda \upsilon$.   Then, probability mass only leaves $[0, \infty)$ via the sink term and does not flow across the boundaries at $\lambda =0$ and $\lambda = \infty$.  Integrating (\ref{Eq:NonlinearSPDEhom}) over $\mathbb{R}^{+}$ and using the aforementioned flux condition along with the boundary condition at $\lambda = \infty$, we have that
\begin{equation*}
0=  \alpha \bar{\lambda} \upsilon(t, \lambda = 0) -\frac{1}{2} \sigma^2 \upsilon (t, \lambda = 0) + \beta^C \big{(} \int_0^{\infty} \lambda \upsilon (t, \lambda) d \lambda \big{)} \upsilon (t, \lambda = 0),
\end{equation*}
which is only satisfied by $\upsilon(t, \lambda = 0) = 0$.

This discussion provides a justification for the choice  of the boundary conditions for a homogeneous pool. The treatment of a heterogeneous pool is analogous.

\subsection{Alternative Representation of Limiting Loss}\label{SS:AlternativeRepresentation}
As we shall see below, if the density $\upsilon(t,\lambda)$ has sufficiently fast decay at $\lambda=\infty$ and at $\lambda=0$, then one can justify an alternative representation of the limiting loss (\ref{Eq:LimitDefaultsHom}). 
By integration by parts, we have
\begin{eqnarray}
L_{t}&=&1-\int_{0}^{\infty}\upsilon(t,\lambda)d\lambda\nonumber\\
&=&1-\left[\int_{0}^{\infty}\upsilon(0,\lambda)d\lambda+\int_{0}^{t}\left[\alpha(\lambda-\bar{\lambda})\upsilon(s,\lambda)-\beta^{S}\lambda b_{0}(X_{s})\upsilon(s,\lambda)\right]_{\lambda=0}^{\lambda=\infty}ds\right.\nonumber\\
& &\left.+
\int_{0}^{t}\left[\frac{1}{2}\sigma^{2}\lambda\upsilon_{\lambda}(s,\lambda)+\frac{1}{2}\sigma^{2}\upsilon(s,\lambda)+
\frac{1}{2}\left(\beta^{S}\right)^{2}\lambda^{2}\sigma_{0}^{2}\upsilon_{\lambda}(s,\lambda)+\left(\beta^{S}\right)^{2}\lambda\sigma_{0}^{2}\upsilon(s,\lambda)\right]_{\lambda=0}^{\lambda=\infty}ds\right.\nonumber\\
& &\left.-\int_{0}^{t}\int_{0}^{\infty}\lambda\upsilon(s,\lambda)d\lambda ds-\beta^{C}\int_{0}^{t}\la \iota,\upsilon(s,\cdot)\ra\left[\upsilon(s,\lambda)\right]_{\lambda=0}^{\lambda=\infty}ds-\int_{0}^{t}\left[\beta^{S}\lambda\sigma_{0}(X_{s})\upsilon(s,\lambda)\right]_{\lambda=0}^{\lambda=\infty}dV_{s}\right].\label{Eq:IntegrationByParts}
\end{eqnarray}

Now, observe that
\begin{equation*}
\mu^N_0(\hat \PP) = \frac{1}{N}\sum_{n=1}^N\delta_{\hat \pp^\NN_0}(\hat \PP)\dfi^\NN_0=\frac{1}{N}\sum_{n=1}^N\delta_{\hat \pp^\NN_0}(\hat \PP)\chi_{\{\tau^\NN>0\}}=1.
 \end{equation*}
This implies that
\begin{equation}
\int_{0}^{\infty}\upsilon(0,\lambda)d\lambda=1.\label{Eq:DensityTime0}
\end{equation}
Recall the boundary conditions $\upsilon(t,0) =\lim_{\lambda\nearrow\infty}\upsilon(t,\lambda) =0$. Then, if we assume that for any $t\in\mathbb{R}_{+}$, $\upsilon(t,\cdot)$ and $\upsilon_{\lambda}(t,\cdot)$ decay fast at infinity, in the sense that $\lim_{\lambda\uparrow\infty}\lambda^{2}\upsilon_{\lambda}(t,\lambda)=\lim_{\lambda\uparrow\infty}\lambda\upsilon(t,\lambda)=0$, the integration by parts formula (\ref{Eq:IntegrationByParts}) and equation (\ref{Eq:DensityTime0}) imply
\begin{equation}
L_{t}=1-\int_{0}^{\infty}\upsilon(t,\lambda)d\lambda=
\int_{0}^{t}\int_{0}^{\infty}\lambda\upsilon(s,\lambda)d\lambda ds\label{Eq:LalternativeRepresentation}
\end{equation}
A particularly interesting consequence of (\ref{Eq:LalternativeRepresentation}) is summarized in the following remark.
\begin{remark}\label{R:QisBounded}
 Relation (\ref{Eq:LalternativeRepresentation})  implies that the integral term in (\ref{Eq:NonlinearSPDEhom}) is bounded. Indeed, the  term is $Q(t)=\int_{0}^{\infty} \lambda\upsilon(t,\lambda)d\lambda$. Then,
(\ref{Eq:LalternativeRepresentation}) can be rewritten as
\begin{equation}
\int_{0}^{\infty}\upsilon(t,\lambda)d\lambda+
\int_{0}^{t}Q(s) ds=1\label{Eq:LalternativeRepresentation2}
\end{equation}
Since $\int_0^{\infty} \upsilon(t, \lambda) d \lambda \leq 1$, relation (\ref{Eq:LalternativeRepresentation2}) implies that for any $t\in\R_{+}$ we have $Q(t)<\infty$. It is easy to see that the corresponding conclusion also holds for the heterogeneous pool due to Condition \ref{A:Bounded}.
\end{remark}
Moreover, (\ref{Eq:LalternativeRepresentation}) implies for the rate of change
\begin{equation*}
\dot{L}_{t}=-\int_{0}^{\infty}\upsilon_{t}(t,\lambda)d\lambda=
\int_{0}^{\infty}\lambda\upsilon(t,\lambda)d\lambda=\int_{0}^{\infty}\lambda\bar{\mu}_{t}(d\lambda).
\end{equation*}
Finally, we mention that one can view $(L_{t},\upsilon(t,\lambda))$ as a pair satisfying
\begin{eqnarray*} d\upsilon(t,\lambda) &=& \left\{\genL_1^*\upsilon(t,\lambda) +\genL_3^{*,X_{t}}\upsilon(t,\lambda)\right\}dt + \genL_2^* \upsilon(t,\lambda)dL_{t} + \genL_4^{*,X_{t}}\upsilon(t,\lambda)dV_t\qquad t,\lambda>0\nonumber\\
\upsilon(0,\lambda) &=&  \Lambda_\circ(\lambda)\label{Eq:NonlinearSPDEhomAlternative}\\
L_{t}&=&1-\int_{0}^{\infty}\upsilon(t,\lambda)d\lambda=
\int_{0}^{t}\int_{0}^{\infty}\lambda\upsilon(s,\lambda)d\lambda ds.\nonumber
\end{eqnarray*}

\section{Numerical method and results}\label{S:Numerics}
We develop and implement a method for the numerical solution of the SPDE (\ref{Eq:NonlinearSPDEhom}) governing the limiting portfolio loss (\ref{Eq:LimitDefaultsHom}) in a homogenous pool. We obtain the distribution of the limiting loss $L_t$, which we propose as an approximation to the loss $L^N_t$ when $N$ is large. Numerical results illustrate the method, as well as the accuracy and computational efficiency of the approximation.

\subsection{Numerical Approaches in the Deterministic Case}\label{SS:NumericsDeterministicCase}
In the case of $\beta^S = 0$, (\ref{Eq:NonlinearSPDEhom}) becomes a deterministic, quasi-linear PDE.  For completeness, we outline some numerical methods for this case.

\subsubsection{No Feedback}
If in addition $\beta^C = 0$, the default times $\tau^{N,n}$ are independent and we have the analytic solution $L_t = 1-\int_0^{\infty} \exp(A(t) + B(t) \lambda) \Lambda_{\circ}(\lambda) d \lambda$ where $A(t)$ and $B(t)$ satisfy
\begin{eqnarray*}
B(t) &=& \frac{1}{\sigma^2}\left(\alpha + \gamma \tanh\left( -\frac{1}{2} \gamma t + c_1\right)\right) \notag \\
A(t) &=&  - c_2 d_2 t + \frac{2 d_1 d_2}{\gamma} \log \left( \frac{\cosh(-\frac{1}{2} \gamma t + c_1 )}{\cosh (c_1) } \right),
\end{eqnarray*}
where $\gamma = \sqrt{\kappa^2 + 2 \sigma^2}$, $c_1 = \tanh^{-1} ( \frac{- \alpha}{ \gamma})$, $c_2 = \frac{\alpha}{\sigma^2}$, $d_1 = \frac{\gamma}{\sigma^2}$, and $d_2 = - \alpha \bar{\lambda}$.

\subsubsection{Finite Difference}\label{SS-finiteDiff}
For the deterministic case where $\beta^C > 0$, a finite difference scheme can efficiently solve the PDE.  We devise a scheme which is implicit in the differential operators and explicit in the integral operator.  A predictor-corrector iteration is employed to increase accuracy for the integral term.  This finite difference scheme is second-order accurate.  Let $\Delta$ be the time-step for
the scheme.  Also, denote $\upsilon_{j} = \upsilon(j \Delta, \lambda)$ and let $\upsilon_{j + \frac{m}{k}}$
for $m = 1,2, \ldots, k-1$ be predictor steps.  Formally,
\begin{eqnarray*}
\frac{\upsilon_{j+\frac{1}{k}} - \upsilon_j}{\Delta} &=& \genL [\frac{1}{2} (\upsilon_{j+\frac{1}{k}} + \upsilon_j)]  + \mathcal{I} [\upsilon_j]   \genL_2^* [\frac{1}{2}(\upsilon_{j+ \frac{1}{k}}+ \upsilon_j)], \notag \\
\frac{\upsilon_{j+\frac{2}{k}} - \upsilon_j}{\Delta} &=& \genL [\frac{1}{2} (\upsilon_{j+\frac{2}{k}} + \upsilon_j)]  + \mathcal{I} [\frac{1}{2}(\upsilon_{j+\frac{1}{k}}+ v_j)] \genL_2^* [(\frac{1}{2}(\upsilon_{j+ \frac{2}{k}}+ \upsilon_j)],  \notag \\
&\vdots& \notag \\
\frac{\upsilon_{j+1} - \upsilon_j}{\Delta} &=& \genL[\frac{1}{2} (\upsilon_{j+1} + \upsilon_j)]  + \mathcal{I} [\frac{1}{2}(\upsilon_{j+\frac{k-1}{k}}+ \upsilon_j)] \genL_2^* [\frac{1}{2}(\upsilon_{j+ 1}+ \upsilon_j) ],
\end{eqnarray*}
where $\genL = \genL_1^* + \genL_3^*$ and $\mathcal{I}[ \upsilon(t, \lambda)] = \beta^C \int_0^{\infty} \lambda \upsilon(t, \lambda) d \lambda$.

\subsection{Method of Moments}\label{SS:MOM}
We provide a method for the numerical solution of the SPDE (\ref{Eq:NonlinearSPDEhom}) that applies in the case that $\beta^S\ge 0$.
Suppose that the boundary conditions for the SPDE are $\upsilon(t, \lambda = 0) = 0$ and $\lim_{\lambda \to \infty} \upsilon(t, \lambda) = 0$, as justified in Section \ref{SS:BoundaryConditions} above.  (Note that the latter boundary condition also implies $\lim_{\lambda \to \infty} \upsilon_{\lambda} (t, \lambda) = 0$.)  Furthermore, suppose that for each $k\in\mathbb{N}$, $\lambda^k \upsilon(t, \lambda)$ is integrable on $\mathbb{R}^{+}$, almost surely.   A sufficient condition is that the solution $\upsilon(t,\lambda)$ decays exponentially in $\lambda$; that is,  there exist constants $C_1, C_2 > 0$ such that $\upsilon (t, \lambda) < C_1 e^{-C_2 \lambda}$ almost surely for $t \geq 0$, $\lambda$ $\in$ $\mathbb{R}^{+}$.  (We note that it was shown in Remark \ref{R:QisBounded} that $u_0$ and $u_1$ exist.) Then, the moments $u_k(t) = \int_0^{\infty} \lambda^k \upsilon(t, \lambda) d \lambda$ exist almost surely.  They follow the SDE system
\begin{equation}\label{Eq: momentSDEone}
\begin{aligned}
d u_k(t)  &= \big\{ u_k(t) \big{(} - \alpha k + \beta^S b_0(X_t) k + 0.5 (\beta^S)^2 \sigma_0^2 (X_t) k (k-1)  \big{)}   \\
&+ u_{k-1}(t) \big{(} 0.5 \sigma^2 k(k-1) + \alpha \bar{\lambda} k + \beta^C k u_1(t) \big{)} - u_{k+1}(t) \big\} dt  + \beta^S \sigma_0(X_t) k u_k(t) d V_t, \\
u_k(0) &= \int_0^{\infty} \lambda^k \Lambda_{\circ} (\lambda) d \lambda.
\end{aligned}
\end{equation}
To find $u_k(t)$, multiply (\ref{Eq:NonlinearSPDEhom})  by $\lambda^k$ and integrate by parts over $[0, \infty)$.  Also, use the boundary conditions at $\lambda = 0$ and $\lambda = \infty$.   Note that the limiting loss $L_t = 1 - u_0(t)$. A standard discretization scheme can be used to numerically solve the system (\ref{Eq: momentSDEone}).\footnote{Since the solution $\upsilon$ of the SPDE is nonnegative, the moments should also be nonnegative.  However, due to the time discretization, the moments may become negative when simulated.  This can cause instability in the numerical scheme.  To avoid these problems, one could immediately set a moment to zero if it ever goes negative.  In particular, instability in the higher moments may occur for large $\beta^S$ due to the exponential growth term $\frac{1}{2} (\beta^S)^2 \sigma_0^2(X_t) k(k-1) u_k(t) dt$.  To reduce this instability, one could solve the transformed moments $w_k(t) = \exp(- \frac{1}{2} (\beta_S)^2 k(k-1) \int_0^t \sigma_0^2(X_s) ds )$. The SDEs for $\{ w_k \}_{k=0}^K$ will not have the exponential growth term anymore.  Note also that $w_0(t) = u_0(t)$.}

Moment methods have previously been applied to deterministic PDEs such as the Boltzmann equation, see \citeasnoun{Boltzmann}, for example.   The moment SDE system in our case is not closed since the $k$-th equation introduces the $(k+1)$-th moment.   So, in practice one must perform a truncation at some level $k = K$ where we let $u_{K+1} = u_K$ (that is, we use the first $K+1$ moments).  In the asymptotic time limit of the case where $\beta^S = \beta^C = 0$, the sensitivity of $u_0$ to  $u_{K+1}$ is of the order $1/K!$. The sensitivity in the more general case $\beta^S>0, \beta^C > 0$ is more difficult to analyze. 
Numerical results, reported below, indicate that the convergence in terms of the number of moments is very rapid.

We remark that $\beta^C$ plays a pivotal role in the truncated version of system (\ref{Eq: momentSDEone}).  If $\beta^C = 0$ and appropriate choices are made for the coefficient functions $b_0(\cdot)$ and $\sigma_0(\cdot)$ of the systematic risk, the truncated system satisfies the global Lipschitz condition and we have the standard existence and uniqueness results.  If $\beta^C > 0$, the truncated system is only locally Lipschitz and therefore there exists a unique solution up to a stopping time $\zeta (u_0(0), \ldots, u_K(0), \omega) : \mathbb{R}^{K+1} \times \Omega \to [0, \infty)$.

In addition, the moments can be inverted to yield $\upsilon(t, \lambda)$.  However, although a distribution uniquely determines its moments, the converse may not be true.  See \citeasnoun{InverseMoment} for some numerical methods for moment inversion.

There is an alternative approach to viewing the moment system.  The system (\ref{Eq: momentSDEone}) is driven by a single diffusion, which suggests there should be a canonical form where only one SDE has a diffusion term and the other SDEs only have drift terms.  Define $\eta_k(t)  = X_t - \frac{1}{k \beta^S} \log( u_k(t))$ for $k \geq 1$ and $\beta^S>0$. Then
\begin{eqnarray*}
d u_0(t) &=& - e^{\beta^S (X_t - \eta_1(t))} dt, \notag \\
 d \eta_k(t) &=& \big{\{} b_0(X_t) +  0.5 k \beta^S \sigma_0^2(X_t) -\frac{1}{\beta^S}  \big{[} - \alpha  + \beta^S b_0(X_t) + 0.5 (\beta^S)^2 \sigma_0^2 (X_t)  (k-1)  \notag \\
&+& e^{(k-1) \beta^S (X_t - \eta_{k-1}(t)) - k \beta^S(X_t - \eta_k(t))} \big{(} 0.5 \sigma^2 (k-1) + \alpha \bar{\lambda} + \beta^C  e^{ \beta^S ( X_t - \eta_1(t))} \big{)} \\
&-& k^{-1} e^{(k+1) \beta^S (X_t - \eta_{k+1}(t)) - k \beta^S(X_t - \eta_k(t))}  \big{]}  \big{ \} } dt. \notag
\end{eqnarray*}
The moment $u_0$ and the ``canonical'' moments $\eta_k$ solve a system of random ODEs. They depend on the path of the systematic risk $X$. A skeleton of $X$ can be generated exactly (without discretization bias) using the methods of \citeasnoun{Beskos}, \citeasnoun{chen}, or \citeasnoun{smelov}.


\subsection{Behavior of Limiting Loss Distribution}\label{SS:AsymptoticLoss}
\begin{figure}[t]
\begin{center}
\includegraphics[scale=0.7]{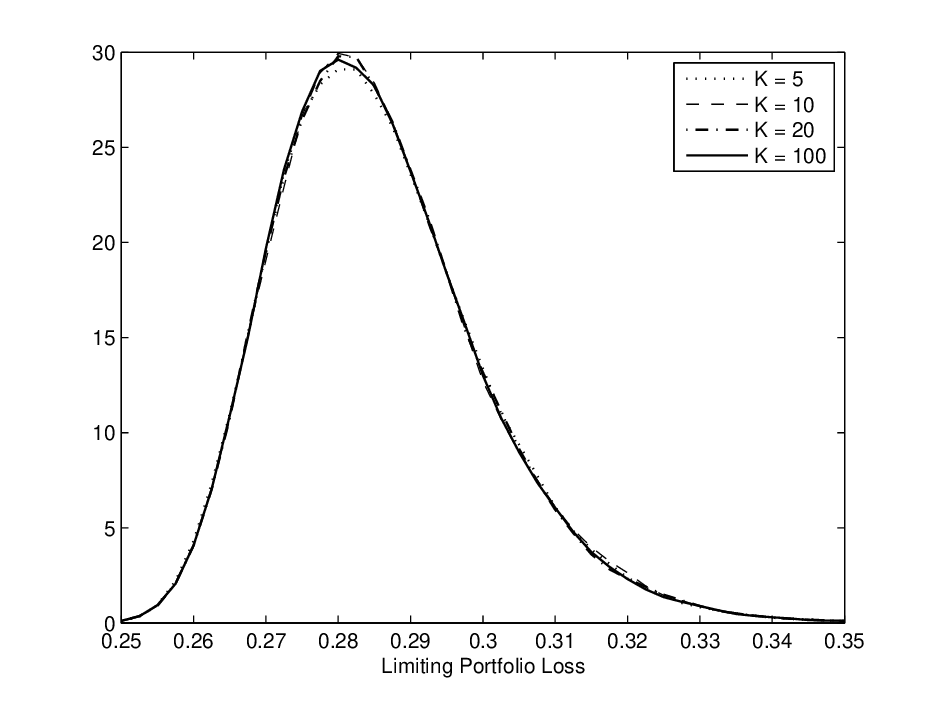}
\end{center}
\caption{\label{fig:truncation} Comparison of distribution of limiting portfolio loss $L_t$ for different truncation levels $K$ at $t=1$.  The parameter case is $\sigma = .9$, $\alpha = 4$, $\bar{\lambda} = .2$, $\lambda_0 = .2$,  $\beta^C = 2$, and $\beta^S = 2$. $50,000$ Monte Carlo trials were used.}
\end{figure}

We provide some numerical results. Here and below, we choose the systematic risk process to be the CIR process  $d X_t = \kappa ( \theta - X_t) dt + \epsilon \sqrt{X_t} d V_t$  where $\kappa = 4$, $\theta = .5$, $\epsilon  = .5$, and $X_0 = .5$.  We choose the initial condition $\Lambda_{\circ}(\lambda) = \delta(\lambda - \lambda_0)$.  The moment system (\ref{Eq: momentSDEone}) is solved using an Euler scheme with time-step of $.01$.  Figure \ref{fig:truncation} shows the rapid convergence of the moment system solution.  Even using as few as six moments ($K=5$), one can achieve a very accurate distribution for the limiting loss $L_t$.  Also, it is noteworthy that in the deterministic case of $\beta^S  = 0$, the moment method is faster for the same accuracy than the finite difference approach outlined in Section \ref{SS-finiteDiff}.

We report some salient features of the limiting loss distribution.  Figure \ref{fig: BetaCVaries} shows the effect of the feedback sensitivity parameter $\beta^C$ on the distribution of $L_t$.   As $\beta^C$ increases, the mean of the losses increases and a heavy tail develops on the right (indicating a greater probability of extreme losses).   Larger $\beta^C$ also causes a wider or more spread-out distribution, indicating a higher variance.  An important ramification is that greater connectivity between firms, modeled here through a nonlinear term, can increase the volatility of their ensemble behavior.  Similarly, increasing the systematic risk sensitivity parameter $\beta^S$ causes heavy tails on the right, see Figure \ref{fig: BetaSVaries}.  A hypothesis is that the initial losses can be sparked by the systematic risk factor (whose influence is determined by the parameter $\beta^S$) and then are later magnified by the contagion risk factor (determined by the parameter $\beta^C$).  This has significant economic implications for the spread of risk through macro credit markets.  The joint presence and interaction of systematic and contagion risk greatly magnifies the likelihood of extreme default events.

To shed more light on this issue, we calculate the Spearman correlation between $X_t$ and $L_t$.  Figure \ref{VaryingBetaC} shows that the correlation over a short time increases with the parameter $\beta^C$ while holding $\beta^S$ fixed.  This indicates that the larger the exposure of firms in the pool to contagion effects, the more susceptible the system is to shocks from the systematic risk. This relationship is demonstrated in Figure \ref{fig: BetaCVaries} by the loss distribution's heavy right tails for large $\beta^C$.  Our finding quantifies the central feature of the model: the complex interaction between systematic risk and contagion.  The system becomes increasingly vulnerable to stresses from the systematic risk as the contagion channel in the system becomes stronger.   


Another interesting observation is that the loss distribution can demonstrate non-monotonic behavior in time; for example, the distribution of losses may widen and then later tighten. Figure \ref{fig: ThreeD} shows the evolution of the distribution of $L_t$ over time $t$, demonstrating this non-monotonicity.  It is noteworthy that our numerical method yields the limiting loss distribution for all horizons simultaneously; this is useful for some applications, including the analysis of portfolio risk measures such as value at risk.  

\begin{figure}[t]
\begin{center}
\includegraphics[scale=0.7]{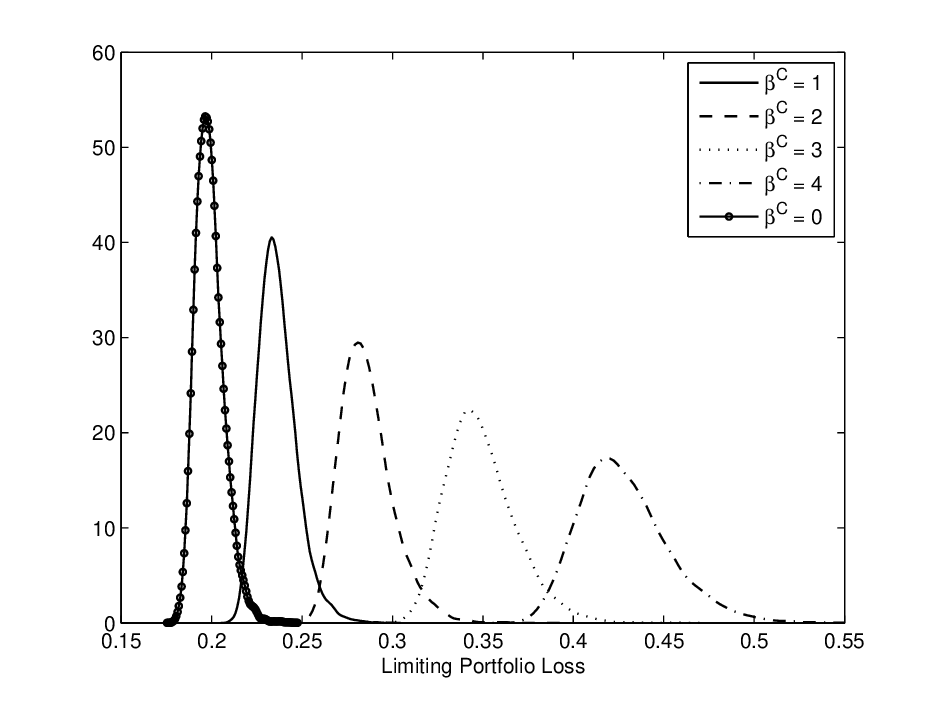}
\end{center}
\caption{\label{fig: BetaCVaries} Comparison of distribution of limiting portfolio loss $L_t$ for different values of the contagion sensitivity $\beta^C$ at $t = 1$.  The parameter case is $\sigma = .9$, $\alpha = 4$, $\bar{\lambda} = .2$, $\lambda_0 = .2$,  and $\beta^S = 2$.  $15,000$ Monte Carlo trials and $16$ moments ($K=15$) were used.  }
\end{figure}

\begin{figure}[t]
\begin{center}
\includegraphics[scale=0.7]{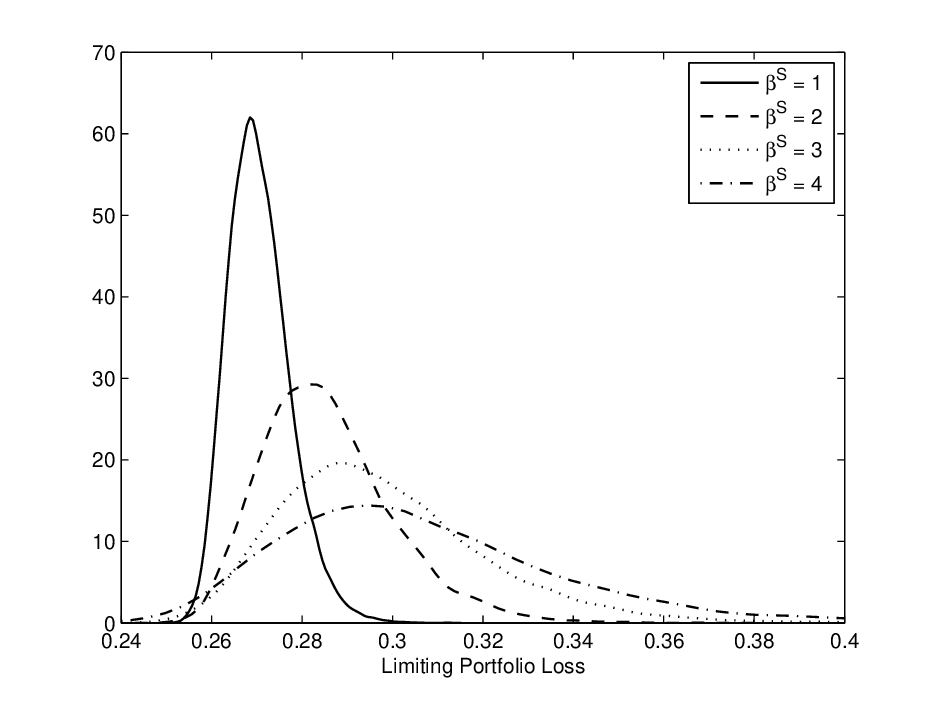}
\end{center}
\caption{\label{fig: BetaSVaries} Comparison of distribution of limiting portfolio loss $L_t$ for different values of the systematic risk sensitivity $\beta^S$ at $t = 1$.  The parameter case is $\sigma = .9$, $\alpha = 4$, $\bar{\lambda} = .2$, $\lambda_0 = .2$,  and $\beta^C = 2$. $15000$ Monte Carlo trials and $16$ moments ($K=15$) were used.  }
\end{figure}

\begin{figure}[t]
\begin{center}
\includegraphics[scale=0.7]{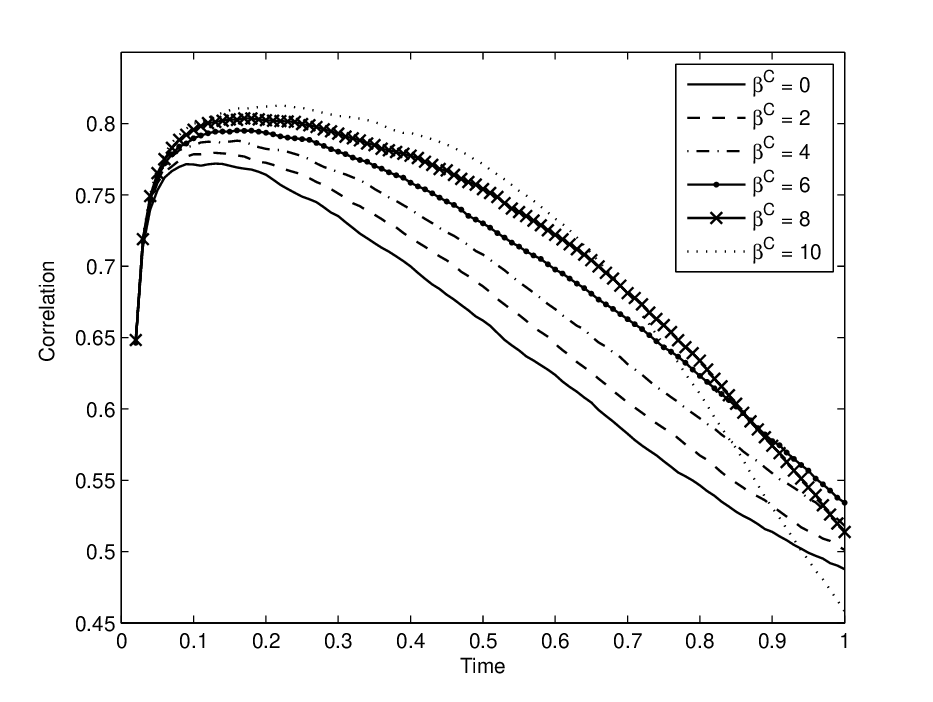}
\end{center}
\caption{\label{VaryingBetaC} Evolution of the Spearman correlation between $X_t$ and $L_t$ over time $t$.  The parameters are $\sigma = .9$, $\alpha = 4$, $\bar{\lambda} = .1$, $\lambda_0 = .1$, and $\beta^S = 4$. $200,000$ Monte Carlo trials and 16 moments ($K=15$) were used.}
\end{figure}


\begin{figure}[t]
\begin{center}
\includegraphics[scale=0.8]{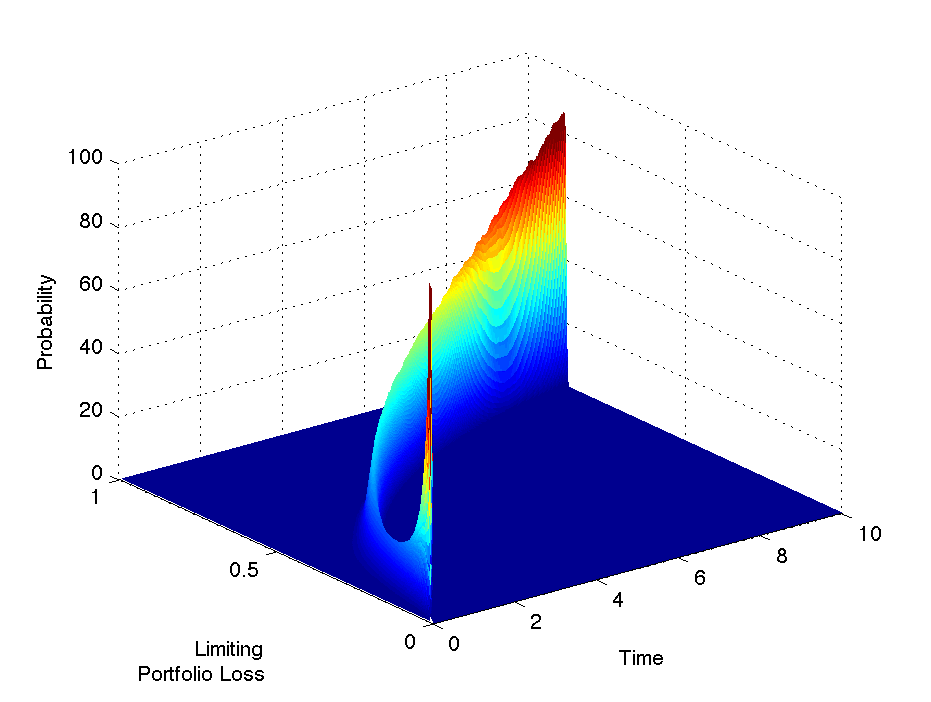}
\end{center}
\caption{\label{fig: ThreeD} Evolution of distribution of limiting portfolio loss $L_t$ over time $t$.  The parameter case is $\sigma = .9$, $\alpha = 4$, $\bar{\lambda} = .2$, $\lambda_0 = .2$, $\beta^S = 4$,  and $\beta^C = 2$. $15,000$ Monte Carlo trials and 16 moments ($K=15$) were used.  }
\end{figure}

\subsection{Accuracy of Large-Portfolio Approximation}\label{SS:ConvergenceOfFiniteSystem}
We analyze the accuracy of the approximation (\ref{limiting-loss-approx}). To this end, we estimate the distribution of the loss $L^N_t$ in a pool of $N$ names by Monte Carlo simulation of the default times $\tau^\NN$. The simulation uses a time-scaling method,
which is based on a discretization of the intensities $\lambda^\NN$ on a time grid $\{t_j \}_{j=1}^J$ where $t_j = j \Delta$ for some $\Delta>0$.  Intensities are simulated on $[t_{j}, t_{j+1})$ using a truncated Euler scheme so that they remain nonnegative.  Firm defaults (as well as jumps in intensity due to defaults) occur at the grid points according to a discretized version of (\ref{E:tau}):
\begin{itemize}
\item Generate independent $\ee_n\sim\mbox{Exp}(1)$ for each firm $n = 1,2, \ldots, N$,
\item \textrm{For} $j = 0,1,2, \ldots, J-1 :$
           \begin{enumerate}
            \item $ \tilde{\lambda}_{t_{j+1}}^{N,n} = \max \big[ 0, \alpha (\bar{\lambda} - \lambda^{N,n}_{t_j}) \Delta + \sigma (\lambda_{t_j}^{N,n})^{1/2} \sqrt{\Delta} \mathcal{N}(0,1) + \beta^S \lambda_{t_j}^{N,n} (X_{t_{j+1}}-X_{t_j})\big]$,
\item \textrm{If} $\tau^{N,n} > t_j$ \textrm{and}  $\Delta( \sum_{i=1}^j \lambda_{t_i}^{N,n} + \tilde{\lambda}_{t_{j+1}}^{N,n}) \geq \ee_n$, \textrm{then} $\tau^{N,n} = t_{j+1}$,
\item $\lambda^{N,n}_{t_{j+1}} = \tilde{\lambda}^{N,n}_{t_{j+1}} + \beta^C \frac{1}{N} \sum_{n=1}^N \chi_{ \{\tau^{N,n} = t_{j+1} \} }$,
\end{enumerate}
\item \textrm{end}.
\end{itemize}

The increments $X_{t_{j+1}} - X_{t_j}$ can be simulated using an Euler scheme similar to the one described above or by an exact scheme.  For the results presented here, we again choose $X$ to be a CIR process with the same parameters as stated earlier. The CIR process $X$ is simulated using  an Euler scheme (truncated at zero as shown above).  We choose a time-step of $\Delta = .01$.

Convergence of the distribution of $L^N_t$ to that of the limiting loss $L_t$ tends be more rapid when $\beta^S$ is larger.  When the variance of the losses is very small (i.e., the limiting losses are close to a deterministic solution), the convergence rate is slower.  Convergence is most rapid in the tails of the loss distribution.  Figures \ref{fig: BetaSthree} and \ref{fig: BetaCone}  show the convergence of the distribution of $L^N_t$ to that of $L_t$ for two different parameter cases.  Convergence is relatively slow, but does indicate that the asymptotic solution is applicable for a portfolio consisting of several thousand names, a portfolio size not unusual in practice.  Figure \ref{fig: NinetyFivePercentVaR} compares the value at risk (VaR) of $L^N_t$ and $L_t$ at the 95 and 99 percent levels.  The limiting VaR is surprisingly accurate even for moderately sized portfolios or several thousand firms. If $N$ is relatively small, then the VaR of $L_t$ tends to understate the VaR of $L^N_t$. This is because fluctuations of $L^N$ due to the idiosyncratic noise terms in the intensity processes have not completely averaged out for small $N$.


\begin{figure}[t]
\begin{center}
\includegraphics[scale=0.7]{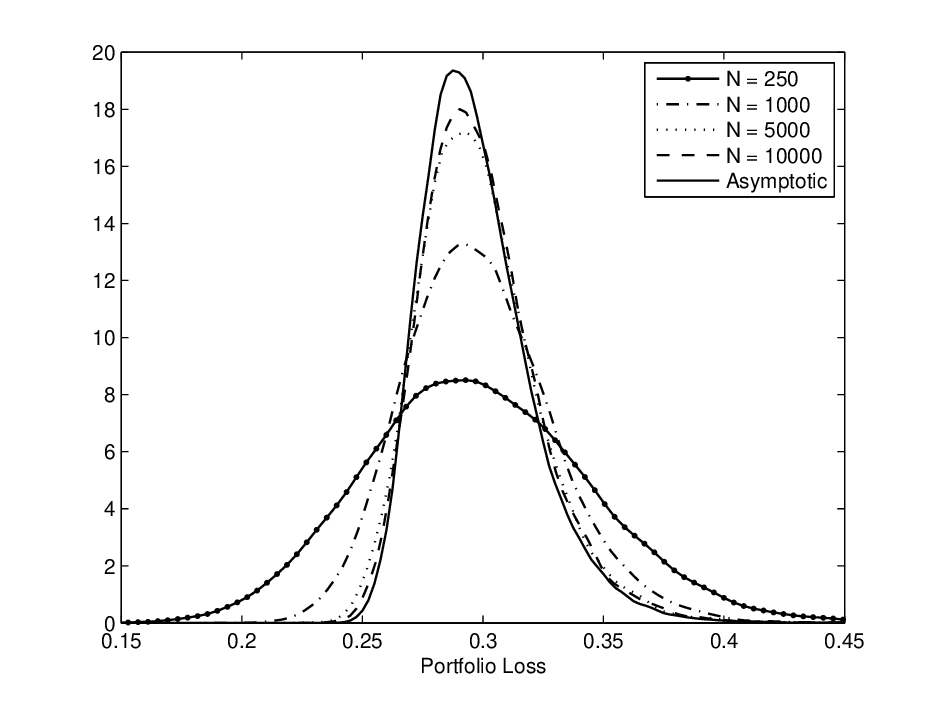}
\end{center}
\caption{\label{fig: BetaSthree} Comparison of distributions of limiting portfolio loss $L_t$ (using $16$ moments) and portfolio loss of finite system $L^N_t$ for different $N$ at $t = 1$.  $25,000$ Monte Carlo trials were used for the finite system and $100,000$ for the asymptotic solution.  The parameter case is $\sigma = .9$, $\alpha = 4$, $\bar{\lambda} = .2$, $\lambda_0 = .2$,  $\beta^C = 2$, and $\beta^S = 3$.  }
\end{figure}

\begin{figure}[t]
\begin{center}
\includegraphics[scale=0.7]{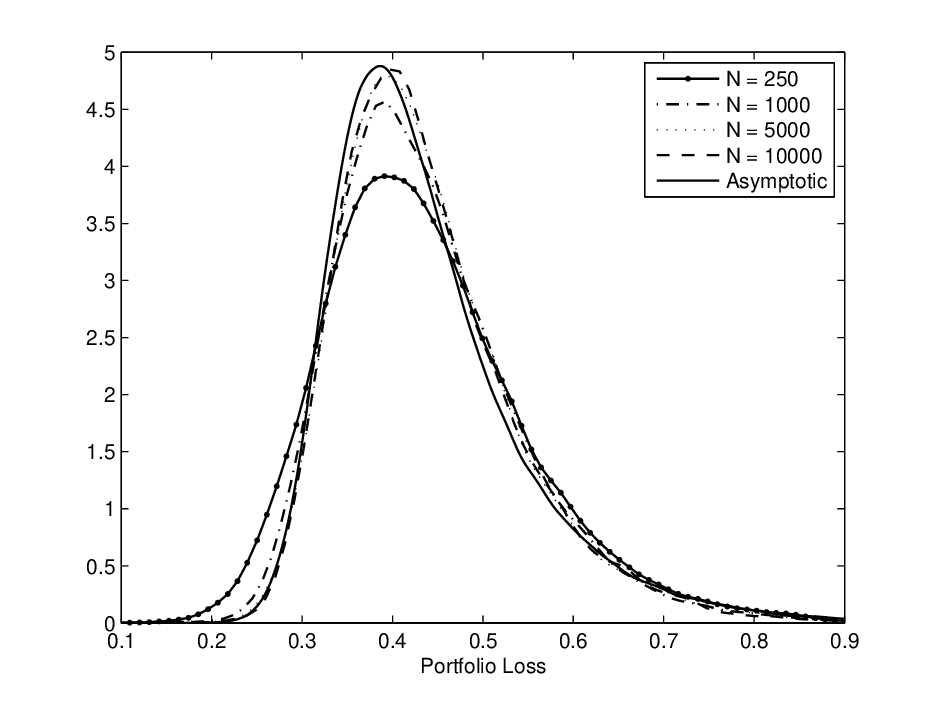}
\end{center}
\caption{\label{fig: BetaCone} Comparison of distributions of limiting portfolio loss $L_t$ (using $16$ moments) and portfolio loss of finite system $L^N_t$ for different $N$ at $t = 1$.     $25,000$ Monte Carlo trials were used for the finite system and $100,000$ for the asymptotic solution.  The parameter case is $\sigma = .9$, $\alpha = 4$, $\bar{\lambda} = .2$, $\lambda_0 = .2$,  $\beta^C =4$, and $\beta^S = 8$.  }
\end{figure}

\begin{figure}[t]
\begin{center}
\includegraphics[scale=0.7]{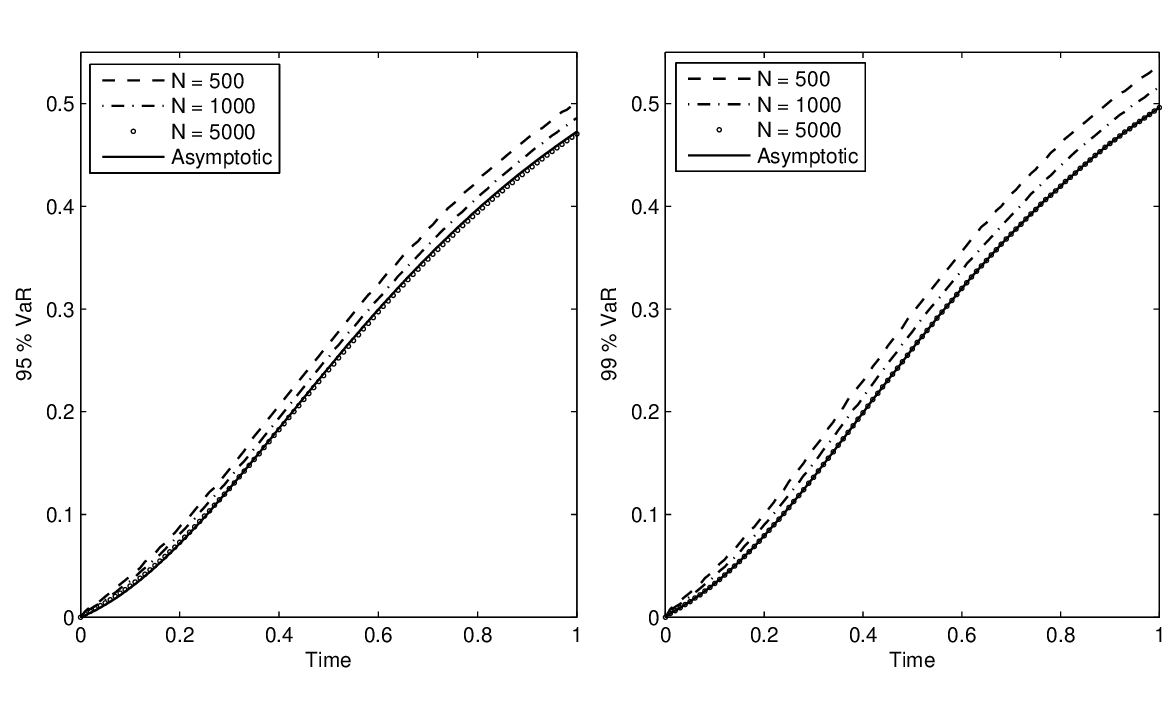}
\end{center}
\caption{\label{fig: NinetyFivePercentVaR} Comparison of $95$ and 99 percent VaR of distributions of limiting portfolio loss $L_t$ (using $16$ moments) and portfolio loss of finite system $L^N_t$ for different $N$ and several horizons $t$. The parameter case is $\sigma = .9$, $\alpha = 4$, $\bar{\lambda} = .2$, $\lambda_0 = .2$, $\beta^S = 2$,  and $\beta^C =4$.  $50,000$ Monte Carlo samples are used for the finite system and $250,000$ for the asymptotic solution.}
\end{figure}


We evaluate the computational efficiency of the approximation.  The parameter case is $\sigma = .9, \alpha = 4, \bar{\lambda} = .2, \beta^C = 2, \beta^S = 3, \lambda_0 = .2$, and $t= 1$.  A discrete time-step of $.01$ is used.  $1000$ Monte Carlo samples are produced. The table below indicates computation times for $L^N_t$, for each of several $N$. The computation time for $L_t$ using $K=200$ is 6.25 seconds.  (Here and below, the computation times are based on a Matlab implementation on a computer running Mac OS X with a $2.7$ GHz dual-core processor.) In general, the computational effort will be of an order $\frac{N}{K}$ greater for the generation of $L^N_t$, where $N$ is the number of firms in the pool and $K$ is the truncation level of the moment method.

\begin{table}[h!]
\begin{center}
\begin{tabular}{| c | c | }
  \hline
\textrm{Portfolio Size $N$}  &  \textrm{Computation time} \\ \hline
  $500$ & 13.58 \textrm{seconds} \\ \hline
  $1000$ & 22.12 \textrm{seconds}  \\ \hline
  $5000$ & 236.24 \textrm{seconds}  \\ \hline
  $10000$ & 441.80 \textrm{seconds} \\ \hline
  $25000$ & 1048.80 \textrm{seconds} \\
  \hline
\end{tabular}
\end{center}
\end{table}

\subsection{Comparison of Method of Moments and Explicit Finite Difference}
An alternate approach to the method of moments is a direct finite difference of the SPDE.  An  implicit method such as Crank-Nicholson cannot be used since the term $\beta^S \sigma_0(X_t) (\lambda \nu)_{\lambda} d V_t$ must be non-anticipating.  We therefore use explicit finite difference, which has the disadvantages of only first-order accuracy in time and conditional stability.

We provide the explicit finite difference scheme in the case of the general diffusion $d X_t = b_0(X_t) dt + \sigma_0(X_t) dV_t$.  The time-step is denoted $\Delta$ and the mesh-size by $\delta$.  Let $\upsilon_{i,j} = \upsilon (i \Delta, j \delta), \lambda_j = j \delta, X_i = X_{i \Delta}$, and $\Delta V_i = V_{i \Delta} -V_{(i-1) \Delta}$, for $i =0,\ldots,N$ and $j = 0, \ldots, J$. Then, the explicit finite difference scheme is
\begin{align*}
\upsilon_{i,j} &= \Delta \big{[} \frac{\mathcal{I}_{i-1}}{2 \delta} - \frac{\sigma^2}{2 \delta} + \frac{\sigma^2 \lambda_j}{2 \delta^2} + (\beta^S \sigma_0(X_{i-1}))^2 \frac{\lambda_j^2}{2 \delta^2} + \frac{\alpha (\bar{\lambda} - \lambda_j)}{2 \delta} + \beta^S \frac{ \lambda_j}{2 \delta} (b_0(X_{i-1})+ \sigma_0(X_{i-1}) \frac{\Delta V_i}{ \Delta} ) \notag \\
&- (\beta^S \sigma_0(X_{i-1}))^2 \frac{\lambda_j}{ \delta} \big{]} \upsilon_{i-1, j-1} \notag \\
&+ \Delta \big{[} \frac{1}{\Delta} + \alpha - \frac{\sigma^2 \lambda_j}{ \delta^2} - (\beta^S \sigma_0(X_{i-1}))^2 \frac{\lambda_j^2}{ \delta^2} - \lambda_j - \beta^S (b_0(X_{i-1}) + \sigma_0(X_{i-1}) \frac{\Delta V_i}{\Delta}) + (\beta^S \sigma_0(X_{i-1}))^2 \big{]} \upsilon_{i-1,j} \notag \\
& +\Delta \big{[} - \frac{\mathcal{I}_{i-1}}{2 \delta} + \frac{\sigma^2}{2 \delta} + \sigma^2 \frac{\lambda_j}{2 \delta^2} + (\beta^S \sigma_0(X_{i-1}))^2 \frac{\lambda_j^2}{2 \delta^2} - \frac{\alpha( \bar{\lambda}- \lambda_j)}{2 \delta} - \beta^S \frac{\lambda_j}{2 \delta} ( b_0(X_{i-1}) +\sigma_0(X_{i-1})  \frac{\Delta V_i}{\Delta} )\notag \\
&+ (\beta^S \sigma_0(X_{i-1}))^2 \frac{\lambda_j}{\delta} \big{]} \upsilon_{i-1, j+1},
\end{align*}
with boundary condition $\upsilon_{i,0} = \upsilon_{i,J} =  0$ and $\mathcal{I}_{i-1}= \sum_{j=1}^J \delta \frac{\upsilon(i-1,j) +\upsilon(i-1,j-1)}{2} $.  For $\sigma_0 = 1$, the criterion for conditional stability for a deterministic diffusion PDE with constant coefficients leads us to propose the approximate criterion $\Delta  \lesssim  \frac{ \delta^2}{(\beta^S \lambda_{max})^2} $ for stability of the above numerical scheme, where $\lambda_{max} = J \delta$.   Numerical studies confirm that this condition for stability is a good approximation.  Note that the time-step must become very small as the effect of systematic risk (i.e., the stochastic terms in the SPDE proportional to the parameter $\beta^C$) increases.  For a general diffusion coefficient $\sigma_0(x)$, we expect instability to generally increase if $\sigma_0(X_t) > 1$ with high probability  (and to decrease if $\sigma_0(X_t) < 1$ with high probability).

Accuracy of explicit finite difference and the method of moments is comparable since we also use a first-order accurate scheme, the Euler method, for the SDE moment system.  However, the great advantage of the SDE moment system is its extremely low computational cost in comparison with explicit finite difference of the SPDE.  If one chooses a mesh with $J$ points for a finite difference scheme, then the finite difference scheme has an order of complexity of at least $J$ coupled SDEs whereas the method of moments can achieve highly accurate results with as little as half a dozen SDEs.  As an example comparison, if the finite difference scheme has a mesh ranging from $0$ to $10$ with a mesh-size of $\delta = .1$, the finite difference scheme has at least the order of complexity of $100$ coupled SDEs.  Furthermore, the explicit finite difference scheme is only conditionally stable.  This means that even if one is satisfied with a time error of $\mathcal{O}(\bar{\Delta})$, one may have to choose a much smaller time-step $\Delta$ to avoid instability.   For instance, if $\delta = .1$, $\beta^S = 5$, and $\lambda_{max} = 10$, our approximate criterion indicates $\Delta$ must be less than $4 \times 10^{-6}$.  Given a desired accuracy in time of $\mathcal{O}(\bar{\Delta})$, we estimate the ratio of the computational cost of the explicit finite difference to that of the method of moments for the case of $\sigma_0 = 1$ to be
\begin{eqnarray}
\frac{\textrm{Cost of Explicit FD}}{\textrm{Cost of Method of Moments}} \approx \frac{J}{K} \min (\bar\Delta (\beta^S \lambda_{\textrm{max}})^2 / \delta^2 , 1).
\end{eqnarray}

As a concrete example, we report computational times for computing the loss at $t = .5$ for the case of $\sigma = 1, \alpha = 4, \bar{\lambda} = 1, \beta^C = 1.5, \beta^S = 2$, and $\lambda_0 = 2$.  The systematic risk $X_t$ is just a Brownian motion (i.e., $b_0 = 0$ and $\sigma_0 = 1$); we can then compare the observed numerical instability with our approximate criterion for stability.  The finite difference method uses a mesh-size of $\delta = .1$ and $\lambda_{max} = 10$ while the method of moments is truncated at level $K = 100$.  This parameter case was chosen to demonstrate the instability of the explicit finite difference scheme when $\beta^S$ becomes reasonably large and the mesh includes large values.  Computational times are reported for $1,000$ Monte Carlo trials.

\begin{table}[h!]
\begin{center}
\begin{tabular}{| c | c | c| }
  \hline
\textrm{Time Step} & \textrm{Method of Moments} & \textrm{Explicit Finite Difference} \\ \hline
$10^{-2}$ & 2.4890 \textrm{seconds} & \textrm{Unstable} \\ \hline
$10^{-3}$ & 25.7241 \textrm{seconds} & \textrm{Unstable} \\ \hline
$10^{-4}$ & 254.6494 \textrm{seconds} & \textrm{Unstable} \\ \hline
$10^{-5}$ & 2561.9614 \textrm{seconds} &  8512.3172 \textrm{seconds} \\ \hline
\end{tabular}
\end{center}
\end{table}

The approximate criterion for stability gives $\Delta \leq 1.1111 \times 10^{-5}$, which matches well with the numerical results.  In this example, we used a rather large number of moments ($100$, to match the number of mesh points in the finite difference scheme).  As we remarked earlier, the computational advantage of the method of moments over finite difference increases substantially if we take a small number of moments since highly accurate results are still achievable even using only a few moments.

\section{Extensions}\label{S:Extensions}
\subsection{Extending the model}\label{SS:StochasticLosses}
We can extend the system (\ref{E:main}) to the case of more general coefficient functions for the intensity as well as stochastic position losses.
Fix $N\in \N$ and $n\in\{1,2,\ldots, N\}$. Let $\{\ell^{\NN}\}_{n=1}^{N}$ be a family of i.i.d. random variables with support $(0,1)$.  The variable $\ell^\NN$ represents  the loss rate at default of the $n$-th name in a pool of size $N$.
Consider the following system:
\begin{equation} \label{E:mainGeneral}
\begin{aligned}
d\lambda^\NN_t &= b(\lambda^\NN_t)dt + \sigma(\lambda^\NN_t)dW^n_t +  \beta^C_\NN \theta(X_{t})dL^N_t+ \beta^S_\NN \gamma(\lambda^\NN_t) dX_t \qquad t>0\\
\lambda^\NN_0 &= \lambda_{\circ, N,n}\\
dX_t &= b_{0}( X_t) dt + \sigma_{0}(X_t)dV_t \qquad t>0\\
X_0&= x_\circ \\
L^N_t &= \frac{1}{N}\sum_{n=1}^N \ell^\NN\chi_{\{\tau^{\NN}\leq t\}} \end{aligned}
\end{equation}
where the coefficient functions $b(\cdot)$, $\sigma(\cdot)$, $\theta(\cdot)$ and $\gamma(\cdot)$ satisfy suitable regularity conditions guaranteeing the existence of a unique nonnegative solution $\lambda_{t}^{\NN}$. 

Define the operators
\begin{equation*}\label{E:Operators2}
\begin{aligned} (\genL_1 f)(\hat \pp) &= \frac12 \sigma^{2}(\lambda)\frac{\partial^2 f}{\partial \lambda^2}(\hat \pp) +b(\lambda)\frac{\partial f}{\partial \lambda}(\hat \pp)-\lambda  f(\hat \pp)\\
(\genL_2 f)(\hat \pp) &= \beta^C \frac{\partial f}{\partial \lambda}(\hat \pp)\\
(\genL_3^{x} f)(\hat \pp) &= \beta^{S}\gamma(\lambda) b_{0}(x)\frac{\partial f}{\partial \lambda}(\hat \pp)+\frac{1}{2}(\beta^{S})^{2}\gamma^{2}(\lambda)\sigma_{0}^{2}(x)\frac{\partial^{2}f}{\partial \lambda^{2}}(\hat \pp)\\
(\genL_4^{x} f)(\hat \pp) &= \beta^{S}\gamma(\lambda)\sigma_{0}(x)\frac{\partial f}{\partial \lambda}(\hat \pp)\\
\QQ(\hat \pp)&=\lambda \ell.
 \end{aligned}
\end{equation*}
Then, following the arguments applied to the system (\ref{E:main}), we can show that the SPDE governing the limiting density takes the following form:
\begin{equation*} d\upsilon(t,\hat \pp) = \left\{\genL_1^*\upsilon(t,\hat \pp) +\genL_3^{*,X_{t}}\upsilon(t,\hat \pp) + \theta(X_{t})\left(\int_{\hat \pp'\in \hat \PP} \QQ(\hat \pp')\upsilon(t,\hat \pp')d\hat \pp'\right) \genL_2^* \upsilon(t,\hat \pp)\right\}dt + \genL_4^{*,X_{t}}\upsilon(t,\hat \pp)dV_t,\label{Eq:NonlinearSPDE1}
\end{equation*}
for $t>0$ and $\hat \pp\in \hat \PP$, where
 $\pp=(\beta^C,\beta^S,\ell)\in \PP\Def \R_+\times \R\times\R$ and $\hat \PP\Def \PP\times \R_+$.

Naturally, one expects that, for large $N$ and for
every $t\geq 0$, $L_{t}^{N}\approx
\ell (1-\int_{0}^{\infty}\upsilon(t,\pp,\lambda)d\lambda) $.

\subsection{Extending the Moment Method}\label{SS:GeneralizationsOfMOM}
The moment SDE system can be extended to the case of general coefficient functions as well as the non-homogeneous parameter case. By Theorem \ref{T:MainLLN0}, in the non-homogeneous case
the SPDE takes the form
\begin{eqnarray}
 d \upsilon(t,\hat \pp ) &=& \Big\{ \frac{1}{2} \sigma^2 (\lambda \upsilon(t,  \hat \pp))_{\lambda \lambda} + \alpha ((\lambda - \bar{\lambda}) \upsilon(t, \hat \pp))_{\lambda} - \lambda \upsilon(t, \hat \pp) - \beta^S b_0(X_t) (\lambda \upsilon(t, \hat \pp))_{\lambda} \notag \\
&+& \frac{1}{2} (\beta^S)^2 \sigma_0^2(X_t) (\lambda^2 \upsilon(t, \hat \pp))_{\lambda \lambda}
- \beta^C \int_{\lambda \in \mathbb{R}^{+}, \pp \in \PP} \lambda \upsilon(t, \hat \pp) d \lambda d \pp (\upsilon(t, \hat \pp))_{\lambda} \Big\} dt \notag\\
& -& \beta^S \sigma_0(X_t) (\lambda \upsilon(t, \hat \pp))_{\lambda} d V_t,  \notag \\
\upsilon(t,\lambda = 0, \pp) &=& \upsilon(t, \lambda = \infty, \pp) = 0, \notag \\
\upsilon(t = 0,\hat \pp) &=& h(\hat \pp). \notag
\label{SPDE}
\end{eqnarray}
The SDE moment system for the non-homogeneous case follows:
\begin{eqnarray*}
d u_k(t, \pp) &=& u_k(t, \pp) [-  \alpha  k+ \beta^S b_0(X_t) k +\frac{1}{2} (\beta^S)^2 \sigma_0^2(X_t) k(k-1)] dt - u_{k+1}(t, \pp) dt \\
&+& u_{k-1}(t, \pp) [ 0.5 \sigma^2 k (k-1) + \alpha \bar{\lambda} k + \beta^C k \int_{p \in P} u_1(t, \pp) d \pp ] dt  + \beta^S \sigma_0(X_t) k u_k (t, \pp) d V_t, \\
u_k(t = 0, \pp) &=& \int_0^{\infty} \lambda^k h(\hat \pp) d \lambda,
\end{eqnarray*}
where $u_k(t, \pp) = \int_{0}^{\infty} \lambda^k u_k(t, \hat \pp) d \lambda$.  The SDE moment system is coupled across the parameter space $\PP$.  For numerical implementation, the parameter space $\PP$ must be discretized.

For the case of intensity processes with general coefficients $b( \cdot )$, $\sigma( \cdot )$, and $\gamma(\cdot)$, we have an SPDE of the form stated in Section \ref{Eq:NonlinearSPDE1}.  A general moment method can be applied to this class of SPDEs, provided that we prescribe $b( \cdot )$, $\sigma( \cdot)$, and $\gamma(\cdot)$ such that the processes $\lambda^{N,n}$ stay positive almost surely.  Also, assume $b( \cdot )$, $\sigma( \cdot)$, and $\gamma(\cdot)$ are analytic on $\mathbb{R}^{+}$.  Then, the generalized moment $$u_{k, k_1, k_2, k_3, k_4, k_5, k_6}(t) = \int_0^{\infty} \lambda^k b(\lambda)^{k_1} \frac{\partial^{k_2} b}{\partial \lambda^{k_2}} \big{(} \sigma(\lambda)^{2} \big{)}^k_3 \frac{\partial^{k_4} \sigma}{\partial \lambda^{k_4}} \gamma(\lambda)^{k_5} \frac{\partial^{k_6} b}{\partial \lambda^{k_6}} \upsilon(t, \lambda) d \lambda$$ solves a moment system similar (albeit more complicated) to (\ref{Eq: momentSDEone}).

\section{Tightness and Identification of the Limit}\label{S:LimitIdentification}

We start by discussing relative compactness of the sequence $\{\mu^N\}_{N\in \N}$.

\begin{lemma}\label{L:MuMeasureTight} The sequence $\{\mu^N\}_{N\in \N}$ is relatively compact in $D_{E}[0,\infty)$.\end{lemma}
\begin{proof}
The proof follows exactly as in Section 6 of \citeasnoun{GieseckeSpiliopoulosSowers2011}. We omit the details.
\end{proof}

Next, we want to use the martingale problem (see
\citeasnoun{MR88a:60130}) to identify the limit of $\mu^N$'s.

Let $\SSS$ be the collection of elements $\Phi$ in $B(\R\times\Pspace(\hat \PP))$ of the form
\begin{equation}\label{E:form} \Phi(x,\mu) = \varphi_{1}(x)\varphi_{2}\left(\la f_1,\mu\ra_E,\la f_2,\mu\ra_E\dots \la f_M,\mu\ra_E\right) \end{equation}
for some $M\in \N$, some $\varphi_{1}\in C^{\infty}(\R)$, some $\varphi_{2}\in C^\infty(\R^M)$ and some $\{f_m\}_{m=1}^M$.  Then $\SSS$ separates $\Pspace(\R \times \hat \PP)$ \citeasnoun{MR88a:60130}.  It thus suffices to show convergence of the martingale
problem for functions of the form \eqref{E:form}.

Let's fix $f\in C^\infty(\hat \PP)$ and understand exactly what happens to $\la f,\mu^N\ra_E$ when one of the firms defaults.  Suppose that the $n$-th firm defaults at time $t$ and that none of the other names defaults at time $t$ (defaults occur simultaneously
with probability zero).  Then
\begin{align*} \la f,\mu^N_t\ra_E &= \frac{1}{N}\sum_{\substack{1\le n'\le N \\ n'\not = n}}  f\left(\pp^{N,n'},\lambda^{N,n'}_t+ \frac{\beta^C_{N,n'}}{N}\right)\dfi^{N,n'}_t\\
\la f,\mu^N_{t-}\ra_E &= \frac{1}{N}\sum_{\substack{1\le n'\le N \\ n'\not = n}}  f\left(\pp^{N,n'},\lambda^{N,n'}_t\right)\dfi^{N,n'}_t +\frac1Nf\left(\pp^\NN,\lambda^\NN_t\right). \end{align*}
Note furthermore that the default at time $t$ means that $\int_{s=0}^t \lambda^\NN_sds=\ee_n$, so $\dfi^\NN_t=0$.
Hence
\begin{equation} \label{E:jjumps} \la f,\mu^N_t\ra_E-\la f,\mu^N_{t-}\ra_E = \jump^f_\NN(t)\end{equation}
where
\begin{equation*} \jump^f_\NN(t) \Def \frac{1}{N}\sum_{n'=1}^N \lb f\left(\pp^{N,n'},\lambda^{N,n'}_t+ \frac{\beta^C_{N,n'}}{N}\right)-f\left(\pp^{N,n'},\lambda^{N,n'}_t\right)\rb\dfi^{N,n'}_t-\frac1Nf\left(\pp^\NN,\lambda^\NN_t\right) \end{equation*}
for all $t\ge 0$, $N\in \N$ and $n\in \{1,2,\dots, N\}$.

For convenience, let's define
\begin{equation*}
\mathcal{G}(f)(x)\Def b_{0}(x)\frac{\partial f}{\partial x}+\frac{1}{2}\sigma_{0}^{2}(x)\frac{\partial^{2}f}{\partial x^{2}}
\end{equation*}
for all $f\in C^{2}(\R)$. This is the generator of the systematic risk.

We now identify the limiting martingale problem for $\mu^N$. For $\hat \pp=(\pp,\lambda)$ where $\pp=(\alpha,\bar \lambda,\sigma,\beta^C,\beta^S)\in \PP$ and $f\in C^\infty(\hat \PP)$, recall the definitions of the operators in (\ref{E:Operators1}).

For $\Phi\in \SSS$ of the form \eqref{E:form} we define the following operators.

\begin{equation}\label{E:limgen}
\begin{aligned}(\genA\Phi)(x,\mu) &\Def (\mathcal{G}(\varphi_{1}))(x)\varphi_{2}\left(\la f_1,\mu\ra_E,\la f_2,\mu\ra_E\dots \la f_M,\mu\ra_E\right)+\nonumber\\
&+\varphi_{1}(x)\sum_{m=1}^M \frac{\partial \varphi_{2}}{\partial x_m}\left(\la f_1,\mu\ra_E,\la f_2,\mu\ra_E\dots \la f_M,\mu\ra_E\right) \lb \la \genL_1f_m,\mu\ra_E + \la \genL_3^{x}f_m,\mu\ra_E + \la \QQ,\mu\ra_E \la \genL_2 f_m,\mu\ra_E \rb\\
&+\frac{\partial \varphi_{1}}{\partial x}(x)\sum_{m=1}^M \frac{\partial \varphi_{2}}{\partial x_m}\left(\la f_1,\mu\ra_E,\la f_2,\mu\ra_E\dots \la f_M,\mu\ra_E\right) \lb \la \sigma_{0}(x)\genL_4^{x}f_m,\mu\ra_E
\rb
 \end{aligned}
\end{equation}
and
\begin{eqnarray}\label{E:limgen2}
(\genB\Phi)(x,\mu) &\Def& \sigma_{0}(x)\frac{\partial \varphi_{1}}{\partial x}(x)\varphi_{2}\left(\la f_1,\mu\ra_E,\la f_2,\mu\ra_E\dots \la f_M,\mu\ra_E\right)\nonumber\\
& &+\varphi_{1}(x)\sum_{m=1}^M \frac{\partial \varphi_{2}}{\partial x_m}\left(\la f_1,\mu\ra_E,\la f_2,\mu\ra_E\dots \la f_M,\mu\ra_E\right) \lb \la \genL_4^{x}f_m,\mu\ra_E  \rb
\end{eqnarray}

Moreover, we define the following processes

\begin{equation}\label{E:limgen3}
\begin{aligned}(\hat{\genA}^{\NN}\Phi)(X_{t},\mu^{N}_{t}) &\Def \varphi_{1}(X_{t})\sum_{n=1}^N \lambda^\NN_t\left\{ \varphi_{2}\left(\la f_1,\mu^N_t\ra_E+\jump^{f_1}_\NN(t),\la f_2,\mu^N_t\ra_E+\jump^{f_2}_\NN(t)\dots \la f_M,\mu^N_t\ra_E+\jump^{f_M}_\NN(t)\right)\right.\\
&\left.\qquad  -\varphi_{2}\left(\la f_1,\mu^N_t\ra_E,\la f_2,\mu^N_t\ra_E\dots \la f_M,\mu^N_t\ra_E\right)\right\} \dfi^\NN_t\\
&-\varphi_{1}(X_t) \sum_{m=1}^M \frac{\partial \varphi_{2}}{\partial x_m}\left(\la
f_1,\mu^N_t\ra_E,\la f_2,\mu^N_t\ra_E\dots \la f_M,\mu^N_t\ra_E\right)  \lb \la \QQ,\mu^N_t\ra_E \la \genL_2 f_m,\mu^N_t\ra_E-\la  \iota f,\mu^N_t\ra_E\rb
 \end{aligned}
\end{equation}
and
\begin{equation}\label{E:limgen4}
\mart^{N}_t \Def  \mart^{N,W}_t + \mart^{N,J}_t
 \end{equation}
where
\begin{equation*}
\mart^{N,W}_t= \frac{1}{N}\sum_{m=1}^{M}\sum_{n=1}^N\int_{0}^{t}\left[\varphi_{1}(X_{r}) \frac{\partial \varphi_{2}}{\partial x_m}\left(\la f_1,\mu^N_r\ra_E,\la f_2,\mu^N_r\ra_E\dots \la f_M,\mu^N_r\ra_E\right) \left\{\sigma_{\NN}\sqrt{\lambda_{r}^{\NN}} \frac{\partial f_m}{\partial \lambda}(\hat \pp^{\NN}_r)\dfi^{N,n}_r\right\}\right] dW^{n}_{r}
\end{equation*}
is the Brownian martingale and $\mart^{N,J}$ is the martingale
\begin{eqnarray*}
\mart^{N,J}_t&=& \sum_{n=1}^{N}\int_{0}^{t}\varphi_{1}(X_{r}) \left\{ \varphi_{2}\left(\la f_1,\mu^N_r\ra_E+\jump^{f_1}_\NN(r),\la f_2,\mu^N_r\ra_E+\jump^{f_2}_\NN(r)\dots \la f_M,\mu^N_r\ra_E+\jump^{f_M}_\NN(r)\right)\right.\\
& &\left.\qquad  -\varphi_{2}\left(\la f_1,\mu^N_r\ra_E,\la f_2,\mu^N_r\ra_E\dots \la f_M,\mu^N_r\ra_E\right)\right\} \left(-d\dfi^\NN_r-\lambda^\NN_r\dfi^\NN_r\right).\end{eqnarray*}
With these definitions we have the following lemma.
\begin{lemma}\label{L:wconv} For any $\Phi\in \SSS$  and any $t>0$ we have that
\begin{equation*}
 \Phi(X_{t},\mu^N_t)=\Phi(X_{0},\mu^N_0)+\int_{0}^t (\genA\Phi)(X_{r},\mu^N_r)dr+\int_{0}^t (\genB\Phi)(X_{r},\mu^N_r)dV_{r}+
 \int_{0}^t (\hat{\genA}^{\NN}\Phi)(X_{r},\mu^{N}_{r})dr+ \mart^{N}_t.\end{equation*}
Moreover, for any $T>0$, the following limits hold
\begin{equation*}
\lim_{N\to \infty}\BE\left[\int_{0}^t \left|(\hat{\genA}^{\NN}\Phi)(X_{r},\mu^{N}_{r})\right|dr\right]=0 \textrm{ and } \lim_{N\to \infty}\sup_{0\leq t\leq T}\BE\left[\mart^{N}_t\right]^{2}=0.
\end{equation*}
\end{lemma}
\begin{proof}
For $\hat \pp=(\pp,\lambda)$ where $\pp=(\alpha,\bar \lambda,\sigma,\beta^C,\beta^S)\in \PP$, define
\begin{align*} (\genL^a f)(\hat \pp) &= \frac12 \sigma^2\lambda\frac{\partial^2 f}{\partial \lambda^2}(\hat \pp) - \alpha(\lambda-\bar \lambda)\frac{\partial f}{\partial \lambda}(\hat \pp)
\end{align*}
Then
$ \genL^a$ is the generator of the idiosyncratic part of the intensity.

We start by writing that
\begin{equation*}\Phi(X_{t},\mu^N_t)=\Phi(X_{0},\mu^N_0)+\sum_{k=1}^{4}\int_{r=0}^t A^{N,k}_r dr+ \sum_{k=1}^{2}\int_{r=0}^t  B^{N,k}_r dV_{r} +\mart_t,\end{equation*}
where $\mart$ is a martingale and
\begin{align*}
 A^{N,1}_t &= \mathcal{G}(\varphi_{1})(X_{t})\varphi_{2}\left(\la f_1,\mu^N_t\ra_E,\la f_2,\mu^N_t\ra_E\dots \la f_M,\mu^N_t\ra_E\right)\\
 A^{N,2}_t &= \varphi_{1}(X_{t})\sum_{m=1}^M \frac{\partial \varphi_{2}}{\partial x_m}\left(\la f_1,\mu^N_t\ra_E,\la f_2,\mu^N_t\ra_E\dots \la f_M,\mu^N_t\ra_E\right)\\
&\qquad \times  \frac{1}{N}\sum_{n=1}^N \lb (\genL^a f_m)(\hat \pp^\NN_t)+ ( \genL_3^{X_t}f_m)(\hat \pp^\NN_t)\rb \dfi^\NN_t\\
&= \varphi_{1}(X_{t})\sum_{m=1}^M \frac{\partial \varphi_{2}}{\partial x_m}\left(\la f_1,\mu^N_t\ra_E,\la f_2,\mu^N_t\ra_E\dots \la f_M,\mu^N_t\ra_E\right)\lb \la   \genL^a f_m,\mu^N_t\ra_E +\la \genL_{3}^{X_t}f_m,\mu^N_t\ra_E\rb  \\
A^{N,3}_t &= \varphi_{1}(X_{t})\sum_{n=1}^N \lambda^\NN_t\lb \varphi_{2}\left(\la f_1,\mu^N_t\ra_E+\jump^{f_1}_\NN(t),\la f_2,\mu^N_t\ra_E+\jump^{f_2}_\NN(t)\dots \la f_M,\mu^N_t\ra_E+\jump^{f_M}_\NN(t)\right)\right.\\
&\qquad \left. -\varphi_{2}\left(\la f_1,\mu^N_t\ra_E,\la f_2,\mu^N_t\ra_E\dots \la f_M,\mu^N_t\ra_E\right)\rb \dfi^\NN_t.\\
A^{N,4}_t &=\frac{\partial \varphi_{1}}{\partial x}(X_{t})\sum_{m=1}^M \frac{\partial \varphi_{2}}{\partial x_m}\left(\la f_1,\mu^N_t\ra_E,\la f_2,\mu^N_t\ra_E\dots \la f_M,\mu^N_t\ra_E\right)\lb \la \sigma_{0}(X_{t}) \genL_{4}^{X_t}f_m,\mu^N_t\ra_E\rb.\\
B^{N,1}_t &= \sigma_{0}(X_{t})\frac{\partial \varphi_{1}}{\partial x}(X_{t})\varphi_{2}\left(\la f_1,\mu^N_t\ra_E,\la f_2,\mu^N_t\ra_E\dots \la f_M,\mu^N_t\ra_E\right)\\
B^{N,2}_t &= \varphi_{1}(X_{t})\sum_{m=1}^M \frac{\partial \varphi_{2}}{\partial x_m}\left(\la f_1,\mu^N_t\ra_E,\la f_2,\mu^N_t\ra_E\dots \la f_M,\mu^N_t\ra_E\right)\lb \la \genL_{4}^{X_t}f_m,\mu^N_t\ra_E\rb.
 \end{align*}

To proceed, let's simplify $\jump^f_\NN$.  For each $f\in C^\infty(\hat \PP)$, $t\ge 0$, $N\in \N$ and $n\in \{1,2,\dots, N\}$, define
\begin{equation}\label{E:effcc0} \tilde \jump^f_\NN(t) \Def \frac{1}{N}\sum_{m=1}^N \frac{\partial f}{\partial \lambda}(\hat \pp^{N,m}_t)  \beta^C_{N,m}  \dfi^{N,m}_t-f\left(\pp^\NN,\lambda^\NN_t\right) =\la \genL_2 f,\mu^N_t\ra_E-f(\hat \pp^N_t). \end{equation}
Then
\begin{equation*} \left|\jump^f_\NN(t)-\frac1N\tilde \jump^f_\NN(t)\right|\le \frac{\KK_{\ref{A:Bounded}}^2}{N^2}\left\|\frac{\partial^2 f}{\partial \lambda^2}\right\|_C, \end{equation*}
where $\KK_{\ref{A:Bounded}}$ is the constant from Condition \ref{A:Bounded}.

Define $\iota(\hat \pp)\Def \lambda$ for $\hat \pp=(\pp,\lambda)\in \hat \PP$.  Setting
\begin{equation*}\label{E:effcc}\begin{aligned}\tilde A^{N,3}_t &\Def \varphi_{1}(X_t)\sum_{m=1}^M \frac{\partial \varphi_{2}}{\partial x_m}\left(\la f_1,\mu^N_t\ra_E,\la f_2,\mu^N_t\ra_E\dots \la f_M,\mu^N_t\ra_E\right)
\frac{1}{N}\sum_{n=1}^N \lambda^\NN_t\tilde \jump^{f_m}_\NN(t) \dfi^\NN_t \\
&=\varphi_{1}(X_t) \sum_{m=1}^M \frac{\partial \varphi_{2}}{\partial x_m}\left(\la
f_1,\mu^N_t\ra_E,\la f_2,\mu^N_t\ra_E\dots \la
f_M,\mu^N_t\ra_E\right)
  \lb \la \QQ,\mu^N_t\ra_E \la \genL_2 f_m,\mu^N_t\ra_E-
\la  \iota f,\mu^N_t\ra_E\rb,
\end{aligned}\end{equation*}
we have that
\begin{equation*} \lim_{N\to \infty}\BE\left[\int_{r=0}^t \left|(\hat{\genA}^{\NN}\Phi)(X_{r},\mu^{N}_{r})\right|dr\right]= \lim_{N\to \infty}\BE\left[\int_{r=0}^t \left|A^{N,3}_r-\tilde A^{N,3}_r\right|dr\right]=0. \end{equation*}
Moreover, by Lemma 3.4 in \citeasnoun{GieseckeSpiliopoulosSowers2011} we have that for any $T>0$ and any $p\geq 1$, there is a constant $C_{0}$ such that $\sup_{\substack{0\le t\le T \\ N\in \N}}\frac{1}{N}\sum_{n=1}^N\BE[|\lambda^\NN_t|^p]<C_{0}$. This and Condition \ref{A:Bounded} imply that
\begin{equation*}
\lim_{N\to \infty}\sup_{0\leq t\leq T}\BE\left[\mart^{N}_t\right]^{2}=0
\end{equation*}

Collecting things together, we get the statements of the lemma.
\end{proof}
\noindent We in particular note the macroscopic effect of the contagion.
\begin{remark}\label{R:effcont} The key step in quantifying the coarse-grained effect of contagion was \eqref{E:effcc0}.  Namely, we average the combination of the jump rate and the exposure to contagion
across the pool.  \end{remark}

\section{Proof of Theorem \ref{T:MainLLN0}}\label{S:MainProof}

Let $\mathbb{Q}_N$ be the $\BP$-law of $(X,\mu^N)$; i.e.,
\begin{equation*} \mathbb{Q}_N(A) \Def \BP\{(X,\mu^N)\in A\} \end{equation*}
for all $A\in \Borel(D_{\R\times E}[0,\infty))$.  Thus $\mathbb{Q}_N\in
\Pspace(D_{\R\times E}[0,\infty))$ for all $N\in \N$. For $\omega\in
D_{\R\times E}[0,\infty)$, define $Y_t(\omega)\Def \omega(t)$ for
all $t\ge 0$.

Also for $\Phi\in \SSS$, define the quantity
\begin{equation}
\Lambda^{\Phi}_{t}(Y)\Def \Phi(Y_t)-\Phi(Y_0)-\int_{r=0}^t (\genA\Phi)(Y_r)dr-\int_{r=0}^t (\genB\Phi)(Y_r)dV_{r}\label{Eq:LimitingMartingale}
\end{equation}

Let  $\mathcal{V}=\bigcup_{t\in\R_{+}}\mathcal{V}_{t}$.
\begin{proposition}\label{P:conv} We have that $\mathbb{Q}_N$ converges \textup{(}in the topology of $\Pspace(D_{\R\times E}[0,\infty))$\textup{)} to the solution $\mathbb{Q}$ of the (filtered) martingale problem for $\Lambda^{\Phi}_{t}(Y)$ by (\ref{Eq:LimitingMartingale}) and such that $\mathbb{Q} Y_0^{-1} = \delta_{x_{\circ}\times\pi\times \Lambda_\circ}$.
In particular, $\mathbb{Q}\{Y_0=x_{\circ}\times\pi\times \Lambda_\circ\}=1$ and
for all $\Phi\in \SSS$ and $0\le r_1\le r_2\dots r_J=s<t<T$ and
$\{\psi_j\}_{j=1}^J\subset B(\R\times E)$, we have that $\Lambda^{\Phi}_{t}$ is a square integrable martingale with respect to both $\filt_{t}\bigvee \mathcal{V}$ and $\filt_{t}$. Namely,
\begin{equation}
\BE\left[\left(\Lambda^{\Phi}_{t}(Y)-\Lambda^{\Phi}_{s}(Y)\right)\prod_{j=1}^J \psi_j(X_{r_j},\mu^N_{r_j})\right]=0 \textrm{ and }
\sup_{0\leq t\leq T}\BE\left[\Lambda^{\Phi}_{t}(Y)\right]^{2}<\infty \label{Eq:LimitingMartingaleProperties}
\end{equation}
Lastly, $\BE[\Lambda^{\Phi}_{t}(Y)\Big|\mathcal{V}]=0$.

\end{proposition}
\begin{proof}
The family $\{\mu^N\}_{N\in \N}$ is relatively compact (as a $D_E[0,\infty)$-valued random variable) by Lemma \ref{L:MuMeasureTight}. Hence $\{X,\mu^N\}_{N\in \N}$ is also relatively compact. Let $(X,\bar{\mu})$ be an accumulation point of one of its convergent subsequences.  Then, $\left(X,\mu^N,\int_{0}^{\cdot}(\genB\Phi)(X_r,\mu^{N}_{r})dV_{r}\right)$ will converge in distribution to  $\left(X,\bar{\mu},\int_{0}^{\cdot}(\genB\Phi)(X_r,\bar{\mu}_{r})dV_{r}\right)$.
This and Lemma \ref{L:wconv} imply that  the process  $\Lambda^{\Phi}_{t}(X,\bar{\mu})$ will satisfy (\ref{Eq:LimitingMartingaleProperties}). Uniqueness of this martingale problem can also be shown as in
\citeasnoun{GieseckeSpiliopoulosSowers2011}.

Of course, we also have that for any $\Phi\in
\SSS$,
\begin{equation*} \BE^{\mathbb{Q}}[\Phi(Y_0)] = \lim_{N\to \infty}\BE^{\mathbb{Q}}[\Phi(X_{0},\mu^N_0)] = \Phi(x_{\circ}\times\pi\times \Lambda_\circ) \end{equation*}
which implies the claimed initial condition. The rest of the statements are easily seen to be true.
\end{proof}
\noindent We next want to identify $\mathbb{Q}$.  This will take a
couple of steps.

For notational convenience we shall write
\begin{equation*}
\BE_{\mathcal{V}}\left[\cdot\right]\Def\BE\left[\cdot\big|\mathcal{V}\right].
\end{equation*}

The next lemma is essential for the characterization of the limit.
Its proof is deferred to Appendix \ref{A:Appendix1}.
\begin{lemma}\label{L:bQDef} Let $W^*$ be a reference Brownian motion and assume that Condition \ref{A:RegularityExogenous} is satisfied.  For each $\hat \pp=(\pp,\lambda_\circ)\in
\hat \PP$ where $\pp = (\alpha,\bar
\lambda,\sigma,\beta^C,\beta^S)$, there is a unique pair
$\{(Q(t),\lambda_{t}(\hat \pp)):t\in[0,T]\}$ taking values in
$\R_+\times\R_{+}$ such that
\begin{equation}\label{E:bQDef} \begin{aligned} Q(t) &= \int_{\substack{\hat \pp=(\pp,\lambda)\in \hat \PP\\
\pp=(\alpha,\bar \lambda,\sigma,\beta^C,\beta^S)}} \BE_{\mathcal{V}}\left\{\lambda^{*}_{t}(\hat\pp)\exp\left[-\int_{s=0}^t \lambda_s^*(\hat \pp)ds\right]
\right\}\pi(d\pp)\Lambda_\circ(d\lambda).
\end{aligned}\end{equation} and
\begin{equation} \lambda^*_t(\hat \pp) = \lambda_\circ - \alpha\int_{s=0}^t (\lambda^*_s(\hat \pp)-\bar \lambda)ds + \sigma\int_{s=0}^t\sqrt{\lambda^*_s(\hat \pp)}dW^*_s + \beta^C \int_{s=0}^t Q(s) ds+\beta^{S}\int_{s=0}^t \lambda^*_s(\hat \pp)dX_{s}. \qquad t\ge 0 \label{E:EffectiveEquation1}
\end{equation}
\textup{(}where $\pi$ and $\Lambda_\circ$ are as in Condition
\ref{A:regularity}\textup{)}.
\end{lemma}

\begin{remark}
For notational convenience we do not write the dependence of $Q$ and
$\lambda^*$ on $X$ but this is always assumed.
\end{remark}

\begin{lemma}\label{L:Qchar} We have that $\mathbb{Q}= \delta_{(X,\bar \mu)}$, where for all $A\in \Borel(\PP)$ and $B\in \Borel(\R_+)$, $\bar \mu$ is given by
\begin{equation*}\label{E:mudef} \bar \mu_t(A\times B) \Def  \int_{\hat \pp= (\pp,\lambda)\in \hat \PP} \chi_A(\pp) \BE_{\mathcal{V}_{t}}\left[\chi_B(\lambda^*_t(\hat \pp))\exp\left[-\int_{s=0}^t \lambda_s^*(\hat \pp)ds\right]\right]
\pi(d\pp)\Lambda_\circ(d\lambda).
\end{equation*}
\end{lemma}
\begin{proof} Recall \eqref{E:EffectiveEquation1} and the operators $\genL_1,\genL_2,\genL^{x}_3,\genL^{x}_4$ from \eqref{E:Operators1} and the definition of $Q$ in \eqref{E:bQDef}.

For any $f\in C^\infty(\hat \PP)$,
\begin{equation*} \la f,\bar \mu_t\ra_E = \int_{\hat \pp= (\pp,\lambda)\in \hat \PP} \BE_{\mathcal{V}_{t}}\left[f(\pp,\lambda_t^*(\hat \pp))\exp\left[-\int_{s=0}^t \lambda_s^*(\hat \pp)ds\right]\right] \pi(d\pp)\Lambda_\circ(d\lambda). \end{equation*}
Using Lemmas \ref{L:AuxiliaryLemma1} and \ref{L:AuxiliaryLemma2} we obtain
\begin{align*} d\la f,\bar \mu_t\ra_E &= \left\{\int_{\hat \pp= (\pp,\lambda)\in \hat \PP} \BE_{\mathcal{V}_{t}}\left[\left[(\genL_1 f)(\pp,\lambda^*_t(\hat \pp))+(\genL^{X_{t}}_3 f)(\pp,\lambda^*_t(\hat \pp))\right]\exp\left[-\int_{s=0}^t \lambda_s^*(\hat \pp)ds\right]\right] \pi(d\pp)\Lambda_\circ(d\lambda)\right\}dt\\
&\qquad + \left\{\int_{\hat \pp= (\pp,\lambda)\in \hat \PP} \BE_{\mathcal{V}_{t}}\left[(\genL_2 f)(\pp,\lambda^*_t(\hat \pp))Q(t)\exp\left[-\int_{s=0}^t \lambda_s^*(\hat \pp)ds\right]\right] \pi(d\pp)\Lambda_\circ(d\lambda)\right\}dt\\
&\qquad + \left\{\int_{\hat \pp= (\pp,\lambda)\in \hat \PP} \BE_{\mathcal{V}_{t}}\left[(\genL_4^{X_{t}} f)(\pp,\lambda^*_t(\hat \pp))\exp\left[-\int_{s=0}^t \lambda_s^*(\hat \pp)ds\right]\right] \pi(d\pp)\Lambda_\circ(d\lambda)\right\}dV_{t}\\
&= \left\{\la \genL_1f,\bar \mu_t\ra_E+ Q(t) \la \genL_2f,\bar \mu_t\ra_E+\la \genL^{X_{t}}_3 f,\bar \mu_t\ra_E\right\}dt+\la \genL^{X_{t}}_4 f,\bar \mu_t\ra_E dV_{t}. \end{align*}

To proceed, define
\begin{equation*} G(t)\Def \int_{\substack{\hat \pp=(\pp,\lambda)\in \hat \PP\\
\pp=(\alpha,\bar \lambda,\sigma,\beta^C,\beta^S)}}  \BE_{\mathcal{V}_{t}}\left[\exp\left[-\int_{s=0}^t \lambda_s^*(\hat \pp)ds\right]\right] \pi(d\pp)\Lambda_\circ(d\lambda). \end{equation*}
On the one hand, we have that
\begin{align*} \dot G(t) &=  -\int_{\substack{\hat \pp=(\pp,\lambda)\in \hat \PP\\
\pp=(\alpha,\bar \lambda,\sigma,\beta^C,\beta^S)}}  \BE_{\mathcal{V}_{t}}\left[\lambda_t^*(\hat \pp)\exp\left[-\int_{s=0}^t \lambda_s^*(\hat \pp)ds\right]\right] \pi(d\pp)\Lambda_\circ(d\lambda)\\
&=  -\int_{\substack{\hat \pp=(\pp,\lambda)\in \hat \PP\\
\pp=(\alpha,\bar \lambda,\sigma,\beta^C,\beta^S)}}  \lambda \bar \mu_t(d\hat \pp) = -\la \QQ,\bar \mu_t\ra_E. \end{align*}
On the other hand, by Lemma \ref{L:bQDef} we have
\begin{equation}\label{E:gg} \dot G(t) = -Q(t). \end{equation}
Thus,  we have that
\begin{equation*} d\la f,\bar \mu_t\ra_E =  \left\{\la \genL_1f,\bar \mu_t\ra_E+ Q(t) \la \genL_2f,\bar \mu_t\ra_E+\la \genL^{X_{t}}_3 f,\bar \mu_t\ra_E\right\}dt+\la \genL^{X_{t}}_4 f,\bar \mu_t\ra_E dV_{t}. \end{equation*}
Thus
\begin{equation*} \Phi(X_{t},\bar \mu_t) = \Phi(X_{0},\bar \mu_0)+\int_{s=0}^t (\genA \Phi)(X_{s},\bar \mu_s)ds+ \int_{s=0}^t (\genB \Phi)(X_{s},\bar \mu_s)dV_{s}, \end{equation*}
and hence $\delta_{\left(X,\bar \mu\right)}$ satisfies the martingale problem generated by $\genA$.  Of course we also have that $\bar \mu_0 = \pi\times \Lambda_\circ$.  By uniqueness, the claim follows.\end{proof}

Now we collect our results to prove the law of large numbers given
in Theorem \ref{T:MainLLN0}.
\begin{proof}[Proof of Theorem \ref{T:MainLLN0}]
In Lemma \ref{L:Qchar} we proved that, for any $f\in C^\infty(\hat \PP)$, the limiting measure $\bar{\mu}$ satisfies the measure evolution equation
\begin{equation*} d\la f,\bar \mu_t\ra_E =  \left\{\la \genL_1f,\bar \mu_t\ra_E+ \la \QQ,\bar \mu_t\ra_E \la \genL_2f,\bar \mu_t\ra_E+\la \genL^{X_{t}}_3 f,\bar \mu_t\ra_E\right\}dt+\la \genL^{X_{t}}_4 f,\bar \mu_t\ra_E dV_{t}. \end{equation*}
From this expression it is immediately derived by integration by
parts that, if there exists a solution to the nonlinear SPDE
(\ref{Eq:NonlinearSPDE}), then the density of $\bar{\mu}$ should
satisfy (\ref{Eq:NonlinearSPDE}).  This concludes the proof of
the theorem.
\end{proof}

\appendix
\section{Proof of Lemma \ref{L:bQDef}}\label{A:Appendix1}
In this section we prove Lemma \ref{L:bQDef}. The proof uses a fixed point theorem argument. For notational convenience we sometimes drop the superscript $*$ and $\hat \pp$ from the notation of $\lambda^*_t(\hat \pp)$ and simply write $\lambda_{t}$.

The square root singularity imposes some technical difficulties in the proof. For this purpose we introduce in Subsection \ref{SS:AuxillaryFcn} an auxiliary  function $\psi_{\eta}(x)$ that will be used later on and study its properties. In Subsection \ref{SS:UncoupledCase} we study existence and uniqueness and properties of $\lambda$ satisfying (\ref{E:EffectiveEquation1}) with a given $Q(t)$.

In Subsection \ref{S:DriftBounded}  we prove the lemma under the additional condition that
$b_{0}$ is bounded. Then, in Subsection \ref{S:DriftGeneral} we prove the lemma using Girsanov's theorem for the $X$ process.

\subsection{An auxiliary function}\label{SS:AuxillaryFcn}
Let $0<\eta\ll 1$ and define
\begin{equation*} \psi_\eta(x) \Def \frac{2}{\ln \eta^{-1}}\int_{y=0}^{|x|}\lb \int_{z=0}^y \frac{1}{z}\chi_{[\eta,\eta^{1/2}]}(z) dz\rb dy \qquad \text{and}\qquad g_\eta(x) \Def |x|-\psi_\eta(x) \end{equation*}
for all $x\in \R$.  We note that $\psi_\eta$ is even, so $g_\eta$ is also even.  Taking derivatives, we have that
\begin{equation*} \dot \psi_\eta(x) = \frac{2}{\ln \eta^{-1}} \int_{z=0}^x \frac{1}{z}\chi_{[\eta,\eta^{1/2}]}(z) dz \qquad \text{and}\qquad \ddot \psi_\eta(x) = \frac{2}{\ln \eta^{-1}} \frac{1}{x}\chi_{[\eta,\eta^{1/2}]}(x) \end{equation*}
for all $x>0$.  Since $\ddot g_\eta=-\ddot \psi_\eta\le 0$, $\dot g_\eta$ is non-increasing.  For $x>\sqrt{\eta}$,
\begin{equation*} \dot g_\eta(x) = 1-2\frac{\ln \eta^{1/2}-\ln \eta}{\ln \tfrac{1}{\eta}}=0,\end{equation*}
so in fact $\dot g_\eta$ is nonnegative on $(0,\infty)$ and it vanishes on $[\sqrt{\eta},\infty)$.  Thus $g_\eta$ is nondecreasing and reaches its maximum at $\sqrt{\eta}$.
Since $g_\eta(0)=0$, we in fact have that
\begin{equation*} 0\le g_\eta(x)\le g_\eta(\sqrt{\eta})\end{equation*}
for all $x\ge 0$.  Since $\dot g_\eta$ is non-increasing on $(0,\infty)$ and $\dot g_\eta(x)=1$ for $x\in (0,\eta)$, we have that $\dot g_\eta(x)\le 1$
for all $x\in (0,\sqrt{\eta})$, so $g_\eta(\sqrt{\eta}) \le \sqrt{\eta}$.  Since $g_\eta$ is even, we in fact must have that $|g_\eta(x)|\le \sqrt{\eta}$ for all $x\in \R$.
Hence
\begin{equation*} |x|\le \psi_\eta(x)+\sqrt{\eta} \end{equation*}
for all $x\in \R$.
We finally  note that
\begin{equation}\left|\ddot \psi_\eta(x)\right|\le \frac{2}{\ln \eta^{-1}}\frac{1}{|x|}\chi_{[\eta,\sqrt{\eta})}(|x|)\le \frac{2}{\ln \eta^{-1}}\min\lb \frac{1}{|x|},\frac{1}{\eta}\rb\label{Eq:SecondDerivativeBound}
 \end{equation}
and that $x\dot{\psi}_{\eta}(x)\geq 0$ for  all $x\in \R$.

\subsection{The uncoupled linear case}\label{SS:UncoupledCase}
Let $\xi$ be a $\{\gilt_t\}_{t\ge 0}$-predictable, nondecreasing, bounded and right-continuous process such that $\xi_{0}=0$.
Consider the SDE
\begin{equation} \label{E:lambdaSDE} \begin{aligned} d\lambda_t &=-\alpha(\lambda_t-\bar \lambda)dt + \sigma \sqrt{\lambda_t\vee 0}dW_t + \beta^C d\xi_t  + \beta^S \lambda_t dX_t \qquad t>0\\
\lambda_0 &= \lambda_\circ.\end{aligned}\end{equation}
Then, as in  Lemmas 3.1 and 3.2 of \citeasnoun{GieseckeSpiliopoulosSowers2011},  we get that (\ref{E:lambdaSDE}) has a unique nonnegative solution such that
$\sup_{t\in [0,T]}\BE[|\lambda_t|^q]<\infty$ for all $T>0$ and $q\ge 1$.

\subsection{Proof assuming that $b_{0}(x)$ is bounded}\label{S:DriftBounded}

In this subsection we prove Lemma \ref{L:bQDef} assuming that $b_{0}(x)$ is bounded. The proof uses a fixed point theorem argument.

The main condition of this subsection is that there exists a $M<\infty$ such that
\begin{equation}
\sup_{x\in\R}|b_{0}(x) |\leq M\label{A:BoundedDrift}
\end{equation}

Let $q\geq 1$ and $S^{q}(\R_{+})$ be the set of $\R_{+}$ valued,
adapted, continuous processes  $\{\lambda_t\}_{t\in[0,T]}$ such that
\begin{equation*}
 \left\Vert \lambda\right\Vert_{T,q}=\left(\sup_{0\leq t\leq T}\BE|\lambda_t|^{q}\right)^{1/q}<\infty
\end{equation*}
The space $S^{q}(\R_{+})$ endowed with the norm $\left\Vert \cdot\right\Vert_{T,q}$ is a Banach space.

Let us consider now a nonnegative process $U_{t}(\hat \pp)\in S^{q}(\R_{+})$. Set
\begin{equation}
\xi(U)_{t}=\int_{\substack{\hat \pp=(\pp,\lambda)\in \hat \PP\\
\pp=(\alpha,\bar \lambda,\sigma,\beta^C,\beta^S)}} \left(1-\BE_{\mathcal{V}}\left\{\exp\left[-\int_{s=0}^t U_s(\hat \pp)ds\right]\right\}\right)\pi(d\pp)\Lambda_\circ(d\lambda).\label{Eq:Mapping1}
\end{equation}
and consider (\ref{E:lambdaSDE}) with $\xi$ in place of $\xi(U)$.  We are going to prove that the map $\Phi$
defined by $\lambda=\Phi(U)$ through (\ref{E:lambdaSDE})-(\ref{Eq:Mapping1}) with $U\in S^{1}(\R_{+})$ is a contraction on $S^{1}(\R_{+})$  equipped with the norm
\begin{equation*}
 \left\Vert \lambda\right\Vert_{t,1}<\infty
\end{equation*}
locally in $t$ which can also be extended to any arbitrary $T$.

We collect in the following lemma some important properties of $\xi$ as defined though (\ref{Eq:Mapping1}).
\begin{lemma}\label{L:PropertiesXi}
Let a nonnegative process $U_{t}(\hat \pp)\in S^{q}(\R_{+})$ be given and let $\KK$ be the constant from Condition \ref{A:Bounded}. The map $\xi(U)_{t}$, as a function  of $t$, is continuous, non-decreasing, positive, bounded uniformly in $t\in\R_{+}$ by $\KK$ and satisfies $\xi(U)_{0}=0$.
\end{lemma}
\begin{proof}
All the statements are obvious.
\end{proof}
Fix $U_{t}(\hat \pp),U'_{t}(\hat \pp)\in S^{1}(\R_{+})$ and write for notational convenience $\xi_{t}=\xi(U)_{t}$ and $\xi'_{t}=\xi(U')_{t}$. The process $Z_{t}=\lambda_{t}-\lambda'_{t}$ satisfies
\begin{equation*}
Z_{t}=-\alpha\int_{0}^{t}Z_{s}ds+\sigma\int_{0}^{t}\left(\sqrt{\lambda_{s}}-\sqrt{\lambda'_{s}}\right)dW_{s} +\beta^C \int_{0}^{t}\left(d\xi_{s}-d\xi'_{s}\right)
+\beta^{S}\int_{0}^{t}Z_{s}dX_{s}
\end{equation*}
Next, we apply It\^{o} formula to $\psi_\eta(x)$ defined in Subsection \ref{SS:AuxillaryFcn} with $x=Z_{t}$, getting
\begin{eqnarray}
\psi_\eta(Z_{t}) &=& -\alpha\int_{0}^{t}Z_{s}\psi'_{\eta}(Z_{s})ds+\int_{0}^{t}\psi'_{\eta}(Z_{s})\left(d\xi_{s}-d\xi'_{s}\right)\nonumber\\
 & & +\frac{\sigma^{2}}{2}\int_{0}^{t}\left(\sqrt{\lambda_{s}}-\sqrt{\lambda'_{s}}\right)^{2}\psi''_{\eta}(Z_{s})ds+\sigma\int_{0}^{t}
 \left(\sqrt{\lambda_{s}}-\sqrt{\lambda'_{s}}\right)\psi'_{\eta}(Z_{s})dW_{s}\nonumber\\
& &+\beta^{S}\int_{0}^{t}\beta_{o}(X_{s})Z_{s}\psi'_{\eta}(Z_{s})ds+\beta^{S}\int_{0}^{t}\sigma_{o}(X_{s})Z_{s}\psi'_{\eta}(Z_{s})dV_{s}+\int_{0}^{t}\left(\beta^{S}\sigma_{o}(X_{s})Z_{s}\right)^{2}\psi''_{\eta}(Z_{s})ds\label{Eq:ItoOnPsi}
\end{eqnarray}

For notational convenience we define the quantity  $\Delta_{t}Y=Y_{t}-Y'_{t}$ for any given couple of stochastic processes $Y,Y'$.

We have the following three technical Lemmas \ref{L:BoundForXi1}-\ref{L:BoundForLambda1}.

\begin{lemma}\label{L:BoundForXi1}
For any $q\geq 1$ and $t>0$ we have
 \begin{equation*}
 \BE \left|\Delta_{t}\xi\right|^{q}\leq \KK^{q}\int_{\hat \pp=(\pp,\lambda)}\BE\int_{0}^{t}\left|\Delta_{s}U(\hat \pp)\right|^{q}ds
 \pi(d\pp)\Lambda_\circ(d\lambda)\leq t\KK^{q}\int_{\hat \pp=(\pp,\lambda)}\left\Vert \Delta_{\cdot}U(\hat \pp)\right\Vert^{q}_{t,q}\pi(d\pp)\Lambda_\circ(d\lambda),
 \end{equation*}
 where $\KK$ is the constant from Condition \ref{A:Bounded}.
\end{lemma}
\begin{proof}
By the definition of $\xi(U)$ from (\ref{Eq:Mapping1}) we get
\begin{equation*}
\Delta_{t}\xi=\int_{\hat \pp=(\pp,\lambda)} \left(\BE_{V}\left\{\exp\left[-\int_{s=0}^t U'_s(\hat \pp)ds\right]-\exp\left[-\int_{s=0}^t U_s(\hat \pp)ds\right]\right\}\right)\pi(d\pp)\Lambda_\circ(d\lambda).
\end{equation*}
Recall the trivial inequality $\left|e^{a}-e^{b}\right|\leq \max\{e^{a},e^{b}\}|a-b|$. Since $U_{s}(\hat{\pp}),U'_{s}(\hat{\pp})\geq 0$, we get that
\begin{equation*}
\left|\exp\left[-\int_{s=0}^t U'_s(\hat \pp)ds\right]-\exp\left[-\int_{s=0}^t U_s(\hat \pp)ds\right]\right|\leq \int_{0}^{t}\left|U_{s}(\hat{\pp})-U'_{s}(\hat{\pp})\right|ds
\end{equation*}
 The latter, H\"{o}lder inequality and Fubini's Theorem imply
\begin{eqnarray}
 \BE \left|\Delta_{t}\xi\right|^{q}&\leq&\KK^{q}\int_{0}^{t}\int_{\hat \pp=(\pp,\lambda)}\BE|\Delta_{s}U(\hat \pp)|^{q}\pi(d\pp)\Lambda_\circ(d\lambda)ds\nonumber\\
&\leq& t\KK^{q}\int_{\hat \pp=(\pp,\lambda)}\sup_{0\leq s\leq t}\BE|\Delta_{s}U(\hat \pp)|^{q}\pi(d\pp)\Lambda_\circ(d\lambda)\nonumber
\end{eqnarray}
This concludes the proof of the lemma.
\end{proof}

\begin{lemma}\label{L:BoundForPsi1}
The following bound holds
\begin{equation}
\BE\psi_\eta(Z_{t})\leq C_{1}\int_{0}^{t}\BE|Z_{s}|ds+C_{2}\int_{\hat \pp=(\pp,\lambda)}\left\Vert \Delta_{\cdot}U(\hat \pp)\right\Vert_{t,1}\pi(d\pp)\Lambda_\circ(d\lambda)+C_{3}
\end{equation}
where $C_{1}=C_{1}\left(\KK, M\right)$, $C_{2}=C_{2}(\KK,t)$ and $C_{3}(\KK,t,\eta)\downarrow 0$ as $\eta\downarrow 0$. Here $\KK$ is the constant from Condition \ref{A:Bounded}
and $M$ as in (\ref{A:BoundedDrift}).

\end{lemma}
\begin{proof}
The proof of this lemma follows by bounding each term on the right hand side of (\ref{Eq:ItoOnPsi}) separately using the bound from Lemma \ref{L:BoundForXi1} and the bounds for the first and second order derivatives of $\psi_{\eta}$, i.e. $|\psi'_{\eta}(x)|\leq 1$ and (\ref{Eq:SecondDerivativeBound}) respectively.
In particular we have the following.

Taking expected value in (\ref{Eq:ItoOnPsi})  we obtain
\begin{eqnarray}
\BE\psi_\eta(Z_{t}) &= & -\alpha\BE\int_{0}^{t}Z_{s}\psi'_{\eta}(Z_{s})ds
+\BE\int_{0}^{t}\beta^{S}b_{o}(X_{s})Z_{s}\psi'_{\eta}(Z_{s})ds+\beta^C \BE\int_{0}^{t}\psi'_{\eta}(Z_{s})\left(d\xi_{s}-d\xi'_{s}\right)\nonumber\\
& &+\frac{\sigma^{2}}{2}\BE\int_{0}^{t}\left(\sqrt{\lambda_{s}}-\sqrt{\lambda'_{s}}\right)^{2}\psi''_{\eta}(Z_{s})ds+\BE\int_{0}^{t}\left(\beta^{S}\sigma_{o}(X_{s})Z_{s}\right)^{2}\psi''_{\eta}(Z_{s})ds\label{Eq:ExpectedValueItoOnPsi}
\end{eqnarray}
Let us now bound each term on the right hand side of (\ref{Eq:ExpectedValueItoOnPsi}).

Due to the boundedness condition on  $b_{0}$ we have  for the first and second term
\begin{equation}
 -\alpha\BE\int_{0}^{t}Z_{s}\psi'_{\eta}(Z_{s})ds
+\BE\int_{0}^{t}\beta^{S}b_{o}(X_{s})Z_{s}\psi'_{\eta}(Z_{s})ds\leq \left(\alpha+|\beta^{S}|M\right)\int_{0}^{t}\BE\left|Z_{s}\right|ds\label{Eq:Terms1and2}
\end{equation}

Regarding the third term we recall the properties of $\xi$ outlined in Lemma \ref{L:PropertiesXi} and that $|\psi'_{\eta}(x)|\leq 1$ for all $\eta\geq 0$ and $x\in\R$. Then, approximating by simple processes, we get that that there exists a constant $C_{0}$ such that
\begin{equation*}
\left|\BE\int_{0}^{t}\psi'_{\eta}(Z_{s})\left(d\xi_{s}-d\xi'_{s}\right)\right|\leq C_{0}\sup_{0\leq s\leq t}\BE\left|\xi_{s}-\xi'_{s}\right|
\end{equation*}
The latter display and Lemma \ref{L:BoundForXi1} imply that
\begin{equation}
\left|\BE\int_{0}^{t}\psi'_{\eta}(Z_{s})\left(d\xi_{s}-d\xi'_{s}\right)\right|\leq C_{0}t\KK\int_{\hat \pp=(\pp,\lambda)}\left\Vert \Delta_{\cdot}U(\hat \pp)\right\Vert_{t,1}\pi(d\pp)\Lambda_\circ(d\lambda)\label{Eq:Terms3}
\end{equation}

Relation (\ref{Eq:SecondDerivativeBound}) gives for the fourth term
\begin{equation}
 \frac{\sigma^{2}}{2}\BE\left|\int_{0}^{t}\left(\sqrt{\lambda_{s}}-\sqrt{\lambda'_{s}}\right)^{2}\psi''_{\eta}(Z_{s})ds\right|\leq
\frac{\sigma^{2}}{2}\int_{0}^{t}\left(\BE|Z_{s}|\psi''_{\eta}(Z_{s})\right)ds
\leq
\frac{\sigma^{2}}{\ln\eta^{-1}}\label{Eq:Terms4}
\end{equation}

Relation (\ref{Eq:SecondDerivativeBound}) and the condition $\BE\int_{0}^{t}\left(\sigma_{o}(X_{s})\right)^{2}ds<\infty$  give for the fifth term
\begin{eqnarray}
\BE\left|\int_{0}^{t}\left(\beta^{S}\sigma_{o}(X_{s})Z_{s}\right)^{2}\psi''_{\eta}(Z_{s})ds\right|&\leq&
 \left(\beta^{S}\right)^{2}\left(\frac{\sqrt{\eta}}{\ln \eta^{-1}}t\right)    \BE\int_{0}^{t}\left(\sigma_{o}(X_{s})\right)^{2}ds\label{Eq:Terms5}
\end{eqnarray}

Collecting now terms (\ref{Eq:Terms1and2})-(\ref{Eq:Terms5}) we obtain the statement of the lemma with
\begin{eqnarray}
C_{1}&=&\alpha+|\beta^{S}|M\nonumber\\
C_{2}&=&C_{0}t\KK\nonumber\\
C_{3}&=&\left(\frac{\sigma^{2}}{\ln\eta^{-1}}\right)+\left(\beta^{S}\right)^{2}\left(\frac{\sqrt{\eta}}{\ln \eta^{-1}}t\right)    \BE\int_{0}^{t}\left(\sigma_{o}(X_{s})\right)^{2}ds\nonumber
\end{eqnarray}
\end{proof}

\begin{lemma}\label{L:BoundForLambda1}
 For any $t>0$, there exists a positive $C(t)$ such that
\begin{equation*}
 \left\Vert \Delta_{\cdot}\lambda\right\Vert_{t,1}\leq C(t)
\int_{\hat \pp=(\pp,\lambda)}\left\Vert \Delta_{\cdot}U(\hat{\pp})\right\Vert_{t,1}\pi(d\pp)\Lambda_\circ(d\lambda).
\end{equation*}
Moreover, $C(t)$ is  continuous, increasing in $t$ and satisfies $\lim_{t\downarrow 0}C(t)=0$.
\end{lemma}
\begin{proof}
Applying $|x|\le \psi_\eta(x)+\sqrt{\eta}$ to  $x=Z_{t}$ and using Gronwall's Lemma we obtain
\begin{eqnarray}
 \BE|Z_{t}|&\leq& \left[ C_{2}\int_{\hat \pp=(\pp,\lambda)}\left\Vert \Delta_{\cdot}U(\hat{\pp})\right\Vert_{t,1}\pi(d\pp)\Lambda_\circ(d\lambda)+C_{3}(t,\eta)+\eta^{1/2}+
C_{1}\int_{0}^{t}e^{C_{1}(t-s)}\left(C_{3}(s,\eta)+\eta^{1/2}\right)ds\right.\nonumber\\
& &\left.+C_{1}\int_{0}^{t}e^{C_{1}(t-s)}\int_{\hat \pp=(\pp,\lambda)}C_{2}\left\Vert \Delta_{\cdot}U(\hat{\pp})\right\Vert_{s,1}\pi(d\pp)\Lambda_\circ(d\lambda) ds\right]\nonumber
\end{eqnarray}
where the constants $C_{1},C_{2},C_{3}$ are as in Lemma \ref{L:BoundForPsi1}.
Taking $\eta\downarrow 0$ we obtain
\begin{equation*}
 \BE|Z_{t}|\leq  C(t)\int_{\hat \pp=(\pp,\lambda)}\left\Vert \Delta_{\cdot}U(\hat{\pp})\right\Vert_{t,1}\pi(d\pp)\Lambda_\circ(d\lambda)\label{Eq:CoupledSDE2}
\end{equation*}
where $C(t)=C_{0}\left[1+ e^{C_{1}t}-e^{C_{1}0}\right]t=C_{0}
e^{C_{1}t}t$ for some constant $C_{0}>0$. This concludes the proof
of the lemma.
\end{proof}

We collect the previous results for the proof of Lemma \ref{L:bQDef} in the case of a bounded drift $b_{0}$.

\begin{lemma}\label{L:bQdef_BoundedDrift}
 If $b_{0}(x)$ satisfies (\ref{A:BoundedDrift}), then the statement of Lemma \ref{L:bQDef} is true.
\end{lemma}

\begin{proof}

Using Lemma \ref{L:BoundForLambda1}, a standard Picard iteration procedure shows that there exists a fixed point $\lambda^{*}$ of $\Phi$, i.e. $\lambda^{*}_{t}=\Phi_{t}(\lambda^{*})$, for $t\in[0,t_{1}]$ such that $C(t_{1})<1$.

Let us show that this fixed point is necessarily unique. Indeed, suppose that $\lambda^{*},\lambda^{'}$ are both fixed points. Then, Lemma \ref{L:BoundForLambda1} implies that
\begin{equation*}
 \left\Vert \lambda^{*}(\hat{\pp})-\lambda^{'}(\hat{\pp})\right\Vert_{t,1}\leq C(t)
\int_{\hat \pp=(\pp,\lambda)}\left\Vert \lambda^{*}(\hat{\pp})-\lambda^{'}(\hat{\pp})\right\Vert_{t,1}\pi(d\pp)\Lambda_\circ(d\lambda)\nonumber
\end{equation*}
By integrating, and using the condition $C(t)<1$ for $t\in[0,t_{1}]$, we immediately obtain that
\begin{equation*}
\int_{\hat \pp=(\pp,\lambda)}\left\Vert \lambda^{*}(\hat{\pp})-\lambda^{'}(\hat{\pp})\right\Vert_{t,1}\pi(d\pp)\Lambda_\circ(d\lambda)=0.
\end{equation*}
Therefore, for every $t\in[0,t_{1}]$ and  $\hat \pp\in \hat \PP$ we should have $\lambda^{*}_{t}(\hat{\pp})=\lambda^{'}_{t}(\hat{\pp})$ almost surely. This gives uniqueness.

For the general case, we only have to subdivide the interval $[0,T]$ into a finite number of small intervals.
Thus, there is a unique fixed point to the equation $\lambda^{*}_{t}=\Phi_{t}(\lambda^{*})$ for all $t\in[0,T]$.
Notice that the equation $\lambda^{*}_{t}=\Phi_{t}(\lambda^{*})$ can be trivially written as the pair of the coupled equations (\ref{E:bQDef})-(\ref{E:EffectiveEquation1}).
 Clearly $Q(t)$ is nonnegative since $\lambda^{*}$ is nonnegative. This concludes the proof of the lemma.
\end{proof}

\subsection{Proof of Lemma \ref{L:bQDef}}\label{S:DriftGeneral}
In this section we prove Lemma \ref{L:bQDef}. By Lemma \ref{L:bQdef_BoundedDrift} we know that the statement of Lemma \ref{L:bQDef} is true if $b_{0}(x)$ is bounded.
Our strategy is to apply this result for the case $b_{0}(x)=0$ and then to generalize using Girsanov's Theorem.

Recall the map $\Phi$ defined by $\lambda_{t}=\Phi_{t}(\lambda)$ through (\ref{E:lambdaSDE})-(\ref{Eq:Mapping1}). For notational convenience we shall write
$\lambda^{\circ}$ for $\lambda^{\circ}_{t}=\Phi_{t}(\lambda^{\circ})$ to emphasize the fact that we are considering (\ref{E:lambdaSDE})-(\ref{Eq:Mapping1}) with $b_{0}(x)=0$. We keep the notation
$\lambda$ for the case of a general $b_{0}(x)$.

Let $u(x)$ be such that $\sigma_{0}(x)u(x)=-b_{0}(x)$ and define the quantity
\begin{equation*}
M_{T}=e^{-\int_{0}^{T}u(X_{s})dV_{s}-\frac{1}{2}\int_{0}^{T}\left|u(X_{s})\right|^{2}ds}
\end{equation*}

For given $\lambda_{0}\in S^{1}(\R_{+})$ we define iteratively the sequence
\begin{equation*}
 \lambda_{n+1,t}=\Phi_{t}(\lambda_{n})
\end{equation*}

We shall prove that the sequence $\{\lambda_{n}\}_{n\in\N}$ is a Cauchy sequence in probability uniformly in $t\in[0,T]$. In particular we have the following lemma.

\begin{lemma}\label{L:CauchyInProbability}
Assume that there is a $p>1$ such that $\BE M_{T}^{p}<\infty$ and let $q>1$ be such that $1/p+1/q=1$. For every $\epsilon>0$, there exists $N<\infty$ such that for all $N<n<m$
\begin{equation*}
 \sup_{0\leq t\leq T}\BP\left[\left|\lambda^{m}_{t}-\lambda^{n}_{t}\right|>\delta\right]\leq C_{0}\left(\epsilon/\delta\right)^{1/q}
\end{equation*}
for every $\delta>0$ and for some constant $C_{0}>0$ independent of $n,m$.
\end{lemma}
\begin{proof}
As we saw in the proof of Lemma \ref{L:bQdef_BoundedDrift}, it is enough to consider the equation $\lambda_{t}=\Phi_{t}(\lambda)$. Let $\lambda^{*},\lambda^{'}$ be two solutions of the equation $\lambda_{t}=\Phi_{t}(\lambda)$.

By Girsanov's Theorem on the absolute continuous change of measure in the space of trajectories, H\"{o}lder and Chebychev inequality, we have the
following
\begin{eqnarray}
\sup_{0\leq t\leq T} \BP\left[\left|\lambda^{m}_{t}-\lambda^{n}_{t}\right|>\delta\right]&=&\sup_{0\leq t\leq T}\BE\left[\chi_{\left\{\left|\lambda^{m}_{t}-\lambda^{n}_{t}\right|>\delta\right\}}\right]\nonumber\\
&=& \sup_{0\leq t\leq T}\BE\left[\chi_{\left\{\left|\lambda^{\circ,m}_{t}-\lambda^{\circ,n}_{t}\right|>\delta\right\}}M_{T}\right]\nonumber\\
&\leq& \left(\sup_{0\leq t\leq T} \BP\left[\left|\lambda^{\circ,m}_{t}-\lambda^{\circ,n}_{t}\right|>\delta\right]\right)^{1/q}\left(\BE M^{p}_{T}\right)^{1/p}\nonumber\\
&\leq& \left( \frac{\sup_{0\leq t\leq T}\BE\left[\left|\lambda^{\circ,m}_{t}-\lambda^{\circ,n}_{t}\right|\right]}{\delta}\right)^{1/q}\left(\BE M^{p}_{T}\right)^{1/p}\nonumber
\end{eqnarray}
where $\lambda^{\circ,m},\lambda^{\circ,n}$ satisfy $\lambda^{\circ,m}_{t}=\Phi_{t}(\lambda^{\circ,m})$ and $\lambda^{\circ,n}_{t}=\Phi_{t}(\lambda^{\circ,n})$.

Then Lemma \ref{L:BoundForLambda1} implies that $\sup_{0\leq t\leq T}\BE\left[\left|\lambda^{\circ,m}_{t}-\lambda^{\circ,n}_{t}\right|\right]$ can be made arbitrarily small by choosing
$n,m$ large enough. This together with the boundedness of  $\BE M^{p}_{T}$ conclude the proof of the lemma.
\end{proof}

Now we are in position to prove Lemma \ref{L:bQDef}.
\begin{proof}[Proof of Lemma \ref{L:bQDef}]
By Lemma \ref{L:CauchyInProbability} we know that there is a nonnegative solution to the equation $\lambda_{t}=\Phi_{t}(\lambda)$. Uniqueness is shown in a similar fashion.
This concludes the proof of the Lemma.
\end{proof}

\section{Some auxiliary results}\label{S:AuxiliaryResults}
In this section we prove some auxiliary lemmata that were needed in various places in the paper. These results are of independent interest.

\begin{lemma}\label{L:AuxiliaryLemma1}
If $Y$ is an integrable $\filt_{t}-$measurable random variable, then $\BE\left[Y|\mathcal{V}_{t}\right]=\BE\left[Y|\mathcal{V}\right]$.
\end{lemma}
\begin{proof}
Notice that $\mathcal{V}=\mathcal{V}_{t}\bigvee \mathcal{V}^{+}_{t}$ where
\begin{equation*}
\mathcal{V}^{+}_{t}=\sigma\left(V_{t+s}-V_{t}; s\geq 0\right).
\end{equation*}
Since $V$ is an $(\filt_{t})$-adapted Brownian motion, we get that $\mathcal{V}^{+}_{t}\subset \mathcal{V}$ is independent of $\filt_{t}$. Therefore, we have
\begin{equation*}
\BE\left[Y|\mathcal{V}_{t}\right]=\BE\left[Y\Big|\mathcal{V}_{t}\bigvee \mathcal{V}^{+}_{t}
\right]=\BE\left[Y|\mathcal{V}\right].
\end{equation*}
\end{proof}

\begin{lemma}\label{L:AuxiliaryLemma2}
Consider a c\`{a}dl\`{a}g and $(\filt_{t})$-adapted process $\left\{Y_{t};t\geq 0\right\}$ such that $\BE[\int_{0}^{T}Y^{2}_{t}dt]<\infty$ for every $T\geq 0$. Then we have
 \begin{equation*}
 \BE\left[\int_{0}^{t}Y_{s}dV_{s}\Big|\mathcal{V}_{t}\right]=
\int_{0}^{t} \BE\left[Y_{s}|\mathcal{V}_{s}\right]dV_{s}.
 \end{equation*}
\end{lemma}
\begin{proof}
 Given an arbitrary $u\in L^{\infty}([0,t];\R)$, i.e. $\sup_{0\leq s\leq t}|u_{s}|<\infty$, consider the quantity
 \begin{equation*}
 \mathcal{Z}_{t}=\exp\left(\frac{1}{2}\int_{0}^{t}|u_{s}|^{2}ds+i \int_{0}^{t}u_{s}dV_{s}\right).
 \end{equation*}
It is easy to see that $\mathcal{Z}_{t}$ satisfies the SDE
\begin{equation*}
\mathcal{Z}_{t}=1+i\int_{0}^{t}\mathcal{Z}_{s}u_{s}dV_{s}.
\end{equation*}
Therefore, on the one hand we have
\begin{eqnarray}
\BE\left[\mathcal{Z}_{t}\BE\left[\int_{0}^{t}Y_{s}dV_{s}\Big|\mathcal{V}_{t}\right]\right]&=&
\BE\left[\mathcal{Z}_{t}\int_{0}^{t}Y_{s}dV_{s}\right]=\BE\left[i\int_{0}^{t}\mathcal{Z}_{s}u_{s}Y_{s}ds\right]\nonumber\\
&=&\BE\left[i\int_{0}^{t}\mathcal{Z}_{s}u_{s}\BE\left[Y_{s}\Big|\mathcal{V}_{s}\right]ds\right].\nonumber
\end{eqnarray}
On the other hand,
\begin{equation*}
\BE\left[\mathcal{Z}_{t}\int_{0}^{t}\BE\left[Y_{s}\Big|\mathcal{V}_{s}\right]dV_{s}\right]
=\BE\left[i\int_{0}^{t}\mathcal{Z}_{s}u_{s}\BE\left[Y_{s}\Big|\mathcal{V}_{s}\right]ds\right].
\end{equation*}
Thus we arrive at the equality
\begin{equation*}
\BE\left[\mathcal{Z}_{t}\BE\left[\int_{0}^{t}Y_{s}dV_{s}\Big|\mathcal{V}_{t}\right]\right]=
\BE\left[\mathcal{Z}_{t}\int_{0}^{t}\BE\left[Y_{s}\Big|\mathcal{V}_{s}\right]dV_{s}\right].
\end{equation*}
Since this statement is true for every $u\in L^{\infty}([0,t];\R)$,  we are done.
\end{proof}

\bibliographystyle{jmr}

%
%
%
%
%
%
%
%
%

\end{document}